\documentclass[12pt]{article}
\usepackage{eurosym}
\usepackage[leqno]{amsmath}
\usepackage[amsmath,thmmarks,thref,framed]{ntheorem}

\usepackage{ascii,amsfonts,amssymb,amstext,bbm,calrsfs,color,dsfont,epsfig,keystroke,latexsym,lmodern,psfrag,textcomp,units,url,yfonts}

\usepackage[utf8]{inputenc}

\usepackage{hyperref}
\hypersetup{colorlinks=true}
\hypersetup{urlcolor=black}
\hypersetup{linkcolor=blue}
\hypersetup{citecolor=blue}

\definecolor{black}{cmyk}{1,1,1,1}
\definecolor{blue}{cmyk}{1,1,0,0}

\theoremstyle{break}
\theoremheaderfont{\small\normalfont\sffamily}
\theorembodyfont{\itshape}
\theoremsymbol{\hfill\EOT}
\theoremseparator{:}
\newtheorem{theorem}{Theorem}
\theoremclass{Theorem}
\theoremstyle{plain}
\theoremheaderfont{\sc}
\theoremsymbol{\hfill\EOT}
\theoremseparator{.}
\newshadedtheorem{conjecture}{Conjecture}
\theoremstyle{break}
\theorembodyfont{\slshape}
\theoremsymbol{\hfill\EOT}
\newtheorem{lemma}{Lemma}
\theoremsymbol{\hfill\LF}
\theoremnumbering{arabic}
\newtheorem{proposition}{Proposition}
\theorembodyfont{\slshape}
\theoremsymbol{\hfill\LF}
\theoremnumbering{arabic}
\newtheorem{corollary}{Corollary}
\newshadedtheorem{corshaded}{Corollary}
\theoremstyle{plain}
\theoremheaderfont{\bfseries}
\theorembodyfont{\small\fontfamily{cmbr}\selectfont}
\theoremsymbol{\hfill\SI}
\theoremseparator{.}
\newtheorem{remark}{Remark}
\theoremstyle{break}
\theoremheaderfont{\bfseries}
\theorembodyfont{\small\fontfamily{cmbr}\selectfont}
\theoremsymbol{\hfill\SI}
\theoremseparator{.}

\theoremstyle{break}
\theoremheaderfont{\ttfamily}
\theorembodyfont{\small\sffamily}
\theoremsymbol{\ensuremath{\ast}}
\theoremseparator{.}

\theoremstyle{plain}
\theoremheaderfont{\itshape}
\theorembodyfont{\normalfont}
\theoremsymbol{\hfill\BS}
\theoremseparator{.}
\newtheorem{definition}{Definition}
\theoremstyle{break}
\theoremheaderfont{\normalfont\ttfamily}
\theorembodyfont{\small\fontfamily{cmbr}\selectfont}
\theoremsymbol{\hfill\ENQ}
\newtheorem{notation}{Notation}
\theoremheaderfont{\normalfont\ttfamily}
\theorembodyfont{\small\fontfamily{cmbr}\selectfont}
\theoremsymbol{\hfill\ENQ}

\theoremheaderfont{\normalfont\ttfamily}
\theorembodyfont{\small\fontfamily{cmbr}\selectfont}
\theoremsymbol{\hfill\ENQ}

\theoremheaderfont{\itshape}
\theorembodyfont{\upshape}
\theoremstyle{nonumberplain}
\theoremseparator{.}
\theoremsymbol{\End}
\newtheorem{proof}{Proof}
\def\1{\mathfrak 1}
\def\a{\mathfrak A}
\def\A{\mathcal U}
\def\cA{\mathcal A}
\def\B{\mathrm B}
\def\BL{\mathcal B}

\def\C{C^{*}}
\def\cal{\mathcal{C}}
\def\CAR{\mathrm{CAR}}
\def\CCR{\mathrm{CCR}}
\def\CP{\mathbb{C}}
\def\d{\mathrm d}
\def\e{\mathrm e}
\def\eqt#1{\texorpdfstring{#1}{}}
\def\E{\mathcal E}
\def\fp{\mathfrak{p}(\H,\a)}
\def\G{\mathcal G}
\def\H{\mathcal H}
\def\h{\mathfrak h}
\def\ii{\mathrm i}
\def\I{\mathcal I}
\def\inner#1{\left< #1 \right>}
\def\j{\mathfrak j}
\def\J{\mathfrak J}
\def\mJ{\mathcal J}
\def\L{\mathcal L}
\def\M{\mathcal M}
\def\mS{\mathcal S}
\def\N{\mathbb N}
\def\Nc{\mathcal N}

\def\p{\partial}
\def\P{\mathcal{P}}

\def\Q{\mathcal Q}

\def\R{\mathbb R}

\def\sCAR{\mathrm{sCAR}(\H,\a)}
\def\spec{\mathrm{spec}}

\def\states{\mathfrak E}

\def\tr{\mathrm{tr}}
\def\Tr{\mathrm{Tr}}
\def\V{\mathcal V}
\def\W{\mathcal W}
\def\X{\mathcal X}

\usepackage{times}

\input{tcilatex}

\usepackage{anysize}

\marginsize{2cm}{2cm}{1cm}{1cm}

\begin{document}

\title{Entropy Power Inequality in Fermionic Quantum Computation}
\author{N. J. B. Aza \and D. A. Barbosa T.}
\date{\today}
\maketitle

\begin{abstract}
We study quantum computation relations on unital finite--dimensional $\CAR$ $\C$--algebras. We prove an entropy power inequality (EPI) in a fermionic setting, which presumably will permit understanding the capacities in fermionic linear optics. Similar relations to the bosonic case are shown, and alternative proofs of known facts are given. Clifford algebras and the Grassmann representation can thus be used to obtain mathematical results regarding coherent fermion states.

\noindent \textbf{Keywords:} Quantum Information Theory, Fermionic Gaussian States, Strongly Continuous Semigroups, Entropy Power.\bigskip

\noindent \textbf{AMS Subject Classification:} 37L05, 81P45, 46L57
\end{abstract}
\tableofcontents

\begin{notation}\label{remark constant}
A norm on the generic vector space $\X$ is denoted by $\Vert \cdot \Vert _{\X}$ and the identity map of $\X$ by $\mathbf{1}_{\X}$. The space of all bounded linear operators on $(\X,\Vert \cdot \Vert _{\X}\mathcal{)}$ is denoted by $\BL(\X)$. The unit element of any algebra $\X$ is always denoted by $\1$, provided it exists of course. The set of linear functionals and states will be denoted by $\X^{*}$ and $\states_{\X}$ respectively, while its tracial state will be written as $\tr_{\X}$. The scalar product of any Hilbert space $\X$ is denoted by $\langle \cdot,\cdot\rangle_{\X}$ and $\Tr_{\X}$ represents the usual trace on $\BL(\X)$. 
\end{notation}
\section{Introduction}
Information Theory is one of the paradigmatic examples in the interphase between physics and mathematics. Several important results concerning strict aspects in Information Theory have caught the eye of mathematicians. This was revealed when \emph{Shannon} founded the \emph{Classical Information Theory} (CIT) \cite{shannon01}; among his proposals it was established that the \emph{measure} of information contained in a physical system can be studied via the probability theory framework. In particular, the measure of information contained in a random variable can be described by \emph{Shannon's entropy}. Some results of interest in CIT are the \emph{Classical Young's Inequality}, \emph{de Bruijin identity} and the \emph{Stam inequality}. As stressed by Blachman \cite{blachman}, all these serve as springboards to prove the \emph{Convolution Inequality for Entropy Powers}, which from a physical point of view, is useful to determine information capacities of broadcast channels \cite{konSmi14}, for instance. See below, Expression \eqref{eq:classic_ineq} for a concrete overview. From the mathematical point of view, such \emph{Entropy Power Inequalities} (EPI) are interesting due to their connection with geometrical quantities such as the \emph{Brunn--Minkowski inequalities}, which bound the volume of the set--sum of two compact convex sets in $\R^{d}$, with $d\in\N$. For further details about the importance of EPI in physics and mathematics see \cite{palmargio14} and \cite{gardner02}, respectively.\par
In the \emph{quantum} setting, one replaces probability densities by \emph{density matrices}, defined on some Hilbert space $\H$: $\rho_{\H}\doteq\{\rho\in\BL(\H)\colon\rho>0\text{ with }\Tr_{\H}(\rho)=1\}$. These density matrices are useful to describe a non--commutative probability measure space in quantum systems. Among the recent analogs to classical probabilities and their quantum counterpart the \emph{Fokker--Planck equation} has been stated. Namely, in \cite{carlen14}, Carlen and Maas used the \emph{Clifford $\C$--algebras} as a probability space, and developed a differential calculus to get a \emph{fermionic} Fokker--Plank equation. They gave an ensemble of parallel results between the quantum and classical cases.\par
In this work we prove an EPI for fermion systems, and we provide pivotal identities such as de Bruijin's identity and the Stam inequality for this non--commutative framework. As stressed in the Ph.D. Thesis \cite{NJBAza}, in a fermionic setting, geometric inequalities and their relation with quantum information are unknown. We focus on a quantum information framework in the scope of fermionic quantum information, namely, fermionic linear optics (FLO). The latter refers to \emph{free}--fermion systems under external potentials, which is reminiscent of \emph{simple} conduction electron problems as at the \emph{Anderson model} \cite{Anderson58,klein2007mott,klein2007cond,JMP-autre2,bru2014heat}. Physically, FLO is a limited form of quantum computation that can efficiently simulate classical computers and their study might help at the understanding of quantum channels \cite{bravyi05,bravyi12clas}. As a one of the main contributions of this paper is to construct an unambiguous mathematical structure to study at FQI, which, in turn, is pivotal to prove an EPI at the fermionic setting.\par
\bigskip
In order to state the relevance of the problem, let $(\Omega,\Sigma,\mathfrak{m})$ be a probability measure space. For random variables $X$ and $Y$ in $(\Omega,\Sigma,\mathfrak{m})$, one \emph{clasically} study information quanties such as the ``Shannon (differential) entropy'' $H(X)$ and the ``convolution'' between $X$ and $Y$, $u_{X+Y}$. The latter are defined by \cite{shannon01}
$$
H(X)\doteq-\int u_{X}(x)\ln u_{X}(x)\mathfrak{m}(\d x)\qquad u_{X+Y}(z)\doteq\int u_{Y}(z-x)u_{X}(x)\mathfrak{m}(\d x),
$$
with $z\in\Omega$ and $u_{X}$ the probability density associated to $X$. Note that $u_{X+Y}$ can be recognized as the probability density of the \emph{output} $X+Y$. Then, if $X\doteq\{X_{j}\}_{1\leq j\leq d}$ and $Y\doteq\{Y_{j}\}_{1\leq j\leq d}$ are two sets of random variables with their joint respective probability measures on $\R^{d\in\N}$, the Shannon entropy of the combination of the variables $X$ and $Y$ satisfies the \emph{entropy power inequality} for classical systems 
$$
\e^{\frac{2}{d}H(X+Y)}\geq\e^{\frac{2}{d}H(X)}+\e^{\frac{2}{d}H(Y)}.
$$
In \cite{konSmi14}, K\"{o}nig and Smith improved last inequality in the sense that for $\lambda\in\cal\equiv[0,1]$ they proven the following two \emph{concave} inequalities:
\begin{equation}
\e^{\frac{2}{d}H(X\boxplus_{\lambda}Y)}\geq\lambda_{X}\e^{\frac{2}{d}H(X)}+\lambda_{Y}\e^{\frac{2}{d}H(Y)},\qquad H(X\boxplus_{\lambda}Y)\geq\lambda_{X}H(X)+\lambda_{Y}H(Y),\label{eq:classic_ineq}
\end{equation}
where the symbol $\boxplus_{\lambda}$ refers to the \emph{addition rule} $\boxplus_{\lambda}\colon\R^{d}\times\R^{d}\to\R$ over the probability density space obeying $u_{X\boxplus_{\lambda}Y}\doteq u_{\sqrt{\lambda_{X}}X+\sqrt{\lambda_{Y}}Y}$ with $\lambda_{X}\doteq\lambda$ and $\lambda_{Y}\doteq1-\lambda$, the weights of the sets of random variables $X$ and $Y$, respectively. In order to prove the inequalities \eqref{eq:classic_ineq}, they use the \emph{heat or diffusion (or Gaussian)} equation, which in the current context is described as follows: for $(X,t)\in\R^{d+1}$ with $d\in\N$, the heat or diffusion (or Gaussian) equation describing a scalar field $u(X,t)\colon\R^{d}\times\R^{+}\to\R$ has the form of the \emph{Cauchy problem}\footnote{More generally the heat equation is the simplest example of \emph{parabolic equation} \cite{brezis2010functional}.}
\begin{equation}
\frac{\p u(X,t)}{\p t}=\Delta u(X,t), \quad u_{0}(X)\doteq u(X,0),\label{eq:heat_eq}
\end{equation}
where $u_{0}(X)$ is the \emph{initial data} and $\Delta\doteq\sum\limits_{i=1}^{d}\frac{\p^{2}}{\p x_{i}^{2}}$ is the \emph{Laplacian}. By \cite[Chapter 2, Sect. 2]{EngelNagel} we know that for any $p\in[1,\infty)$ the \emph{diffusion semigroup}, $\{\P_{t}\}_{t\in\R^{+}}\in L^{p}(\R^{d})$, associated to the latter equation verifies for $u_{0}\in L^{p}(\R^{d})$ that
$$
\P_{t}u_{0}(X)=\frac{1}{\left(4\pi t\right)^{\frac{d}{2}}}\int_{\R^{d}}\e^{-\frac{\Vert X-Y\Vert^{2}}{4t}}u_{0}(Y)\d Y.
$$
Then, one can prove that if $\P_{0}=1$, the family of semigroups $\P\doteq\{\P_{t}\}_{t\in\R^{+}}\in L^{p}(\R^{d})$ is strongly continuous and even more, the solution of \eqref{eq:heat_eq} is provided by $u_{0}(X)$ and $\P$ via $u(X,t)=\P_{t}u_{0}(X)$. See Definition \ref{def:semigroup} at Appendix \ref{appen:qin} for a general formulation of semigroups. Note that in \cite{konSmi14}, is assumed that the time dependency of the random set $X$ is $X_{t}= X+\sqrt{t}Z$, for $t\in\R_{0}^{+}$ and $Z\doteq\{Z_{j}\}_{1\leq j\leq d}$ is a set of random variables with \emph{standard normal} distribution, $\Nc(0,1)$.\\
Quantum \emph{bosonic} versions of \eqref{eq:classic_ineq}--\eqref{eq:heat_eq} were given in \cite{konSmi14} and \cite{palmargio14}, and hence the mathematical framework of $\CCR$\footnote{The name refers to \emph{Canonical Commutation Relations}.} algebras was required. They consider bosonic systems \emph{interacting} with the environment such that dissipative processes can occur. The latter means that given some initial \emph{density matrix} $\rho_{0}\equiv\rho$ on a well--defined subset of a $\CCR$ $\C$--algebra $\W$ ($\rho$ satisfies $\tr_{\W}(\rho)=1$ and it is positive), the density matrix $\rho_{t}$ evolves over time via the \emph{Markovian master equation} (see \eqref{eq:qde} and \eqref{eq:Liouvillean})
\begin{equation}
\frac{\d}{\d t}\rho_{t}=\L\rho_{t},\qquad \rho_{t}\doteq\e^{t\L}\rho,\label{eq:qde1} 
\end{equation}
for any $t\in\R_{0}^{+}$. Here, $\L$ is the so--called \emph{Liouvillean} of the system, which is the infinitesimal (unbounded) generator of the strongly continuous semigroup $\P\doteq\{\e^{t\L}\}_{t\in\R_{0}^{+}}$. By assuming that the $\CCR$ algebra has a unit operator $\1$, explicitly, for any $A\in\mathcal{D}(\L)$ (domain of $\L$), $\L$ is given by 
$$
\L A=-\frac{1}{4}\sum_{i\in\I}\left(\left[Q_{i},\left[Q_{i},A\right]\right]+\left[P_{i},\left[P_{i},A\right]\right]\right),
$$
where $\I$ is a finite index set, and $\{Q_{i}\}_{i\in\I}$, $\{P_{i}\}_{i\in\I}$ are two families of operators satisfying the $\CCR$ relations
$$
[Q_{i},P_{j}]=\ii\delta_{i,j}\1,\quad[Q_{i},Q_{j}]=0=[P_{i},P_{j}],
$$
for $i,j\in\I$. As is usual, $Q_{i}$ and $P_{i}$ are the \emph{position} and \emph{momentum} operators at the $i$--\emph{mode}. Here, $[A,B]\doteq AB-BA$ denotes the commutator between $A$ and $B$.\\
Instead of using for $\lambda\in\cal\equiv[0,1]$ the addition rule $\boxplus_{\lambda}$ described above, they consider the \emph{beam splitter} quantum \emph{channel} $\M_{\mathbb{U}_{\lambda}}$ defined by:
\begin{equation}
\rho_{\mathbf{C}}\equiv\M_{\mathbb{U}_{\lambda}}(\rho_{\mathbf{A}}\otimes\rho_{\mathbf{B}})\doteq\tr_{\mathbf{B}}\left(\mathbb{U}_{\lambda}^{*}\left(\rho_{\mathbf{A}}\otimes\rho_{\mathbf{B}}\right)\mathbb{U}_{\lambda}\right),\label{eq:intro_qchannel} 
\end{equation}
where $\rho_{\mathbf{A}}$ and $\rho_{\mathbf{B}}$ are the density matrices of two different \emph{but} similar bosonic systems, e.g., the system $\mathbf{A}$ is described by the $\CCR$ $\C$--algebra $\W_{\mathbf{A}}$ generated by the family of (unbounded) operators $\{a_{i}\}_{i\in\I}$ while the family of (unbounded) operators $\{b_{i}\}_{i\in\I}$ generates the $\CCR$ $\C$--algebra $\W_{\mathbf{B}}$, which describes the bosonic system $\mathbf{B}$. In fact, they both satisfy the $\CCR$ ($\C$--algebra) relations
$$
[a_{i},a_{j}^{*}]=\delta_{i,j}\1,\qquad[b_{i},b_{j}^{*}]=\delta_{i,j}\1, \quad\text{with}\quad i,j\in\I.
$$
Here, as is usual, for $A,B\in\W_{\mathbf{A},\mathbf{B}}$, $[A,B]\doteq AB-BA\in\W_{\mathbf{A},\mathbf{B}}$ denotes the commutator between $A$ and $B$, which by simplicity we assume that is closed on $\W_{\mathbf{A},\mathbf{B}}$. Moreover, for any $a_{i}\in\W_{\mathbf{A}}$ and $b_{j}\in\W_{\mathbf{B}}$ we impose $[a_{i},b_{j}]=0$. In \eqref{eq:intro_qchannel}, for $\lambda\in\cal$, $\mathbb{U}_{\lambda}$ is a unitary operator implementing the beam--splitter \emph{$^{*}$--automorphism} on the \emph{tensor} product algebra $\W_{\mathbf{A}}\otimes\W_{\mathbf{B}}$ so that $\mathbb{U}_{\lambda}\mathbb{U}_{\lambda}^{*}=\1$, and $\tr_{\mathbf{B}}(\cdot)$ means that one of the output \emph{signals} (in this case of the system $\mathbf{B}$) is discarded in order to get the output \emph{mixed} density matrix $\rho_{\mathbf{C}}$. Additionally, under the quantum channel $\M_{\mathbb{U}_{\lambda}}$ the mixing of the families $\{a_{i}\}_{i\in\I}$ and $\{b_{i}\}_{i\in\I}$ provides a new family of (unbounded) operators $\{c_{i}\}_{i\in\I}:$
$$
c_{i}=\sqrt{\lambda_{\mathbf{A}}}a_{i}+\sqrt{\lambda_{\mathbf{B}}}b_{i},\qquad i\in\I,
$$
which uphold the $\CCR$ relations, i.e., $[c_{i},c_{j}]=\delta_{i,j}\1$, with $\lambda_{\mathbf{A}}\doteq\lambda,\lambda_{\mathbf{B}}\doteq1-\lambda$. For more details about channels in the $\C$--algebra context, see Appendix \ref{appen:qin}.\\ 
Then the \emph{modified} bosonic version of the entropy power inequality as given in \eqref{eq:classic_ineq} is the following concave inequality:
$$
\E(\mathbf{C})\geq \lambda_{\mathbf{A}}\E(\mathbf{A})+\lambda_{\mathbf{B}}\E(\mathbf{B}),
$$
where for $\mathbf{D}\in\{\mathbf{A},\mathbf{B},\mathbf{C}\}$, $\rho_{\mathbf{D}}$ is a density matrix, $\E(\mathbf{D})\doteq\e^{s(\rho_{\mathbf{D}})}$ is the \emph{quantum entropy power} with $s(\rho_{\mathbf{D}})\doteq\frac{1}{|\I|}\mS(\rho_{\mathbf{D}})$ the \emph{density entropy} of $\rho_{\mathbf{D}}$ such that
$$
\mS(\rho_{\mathbf{D}})\doteq-\tr_{\W}(\rho_{\mathbf{D}}\ln\rho_{\mathbf{D}}),
$$
denotes the von Neumann entropy.\par
\bigskip

In contrast to the mentioned results, the current work is focused in its fermionic version: Our contributions regard the study of fermionic systems in a quantum computation framework. From the physical point of view, this is of interest because as already mentioned FLO is classically simulable. Mathematically, one must consider a non--commutative framework and hence it is interesting to study the existence of the inequality \eqref{eq:classic_ineq} for the underlying setting.  \\
The standard mathematical formalism when one deals with fermion systems are $\CAR$ $\C$--algebras. If one consider a unital finite $\CAR$ $\C$--algebra $\cA\equiv(\cA,^{*},+,\cdot,\Vert\cdot\Vert_{\cA})$ of size $2^{2N}$ with $N\in\N_{0}$, then this is isomorphic to the $\C$--algebra of the square complex matrices $\mathrm{Mat}(2^{N},\CP)$ \cite{BratteliRobinson}. For a finite Hilbert space $\h$, two very well--studied $\CAR$ $\C$--algebras are that given by (i) the \emph{usual} $\CAR$ algebra $\A\equiv\A(\h)$ generated by the identity $\1$ and the elements $\{a(\varphi)\}_{\varphi\in\h}$ such that satisfy the $\CAR$
\begin{equation}
 a(\varphi_{1})a(\varphi_{2})^{*}+a(\varphi_{2})^{*}a(\varphi_{1})=\inner{\varphi_{1},\varphi_{2}}\1, \qquad \varphi_{1},\varphi_{2}\in\h,\label{eq: CAR basic}
\end{equation}
and (ii) the \emph{Clifford} $\CAR$ algebra $\Q\equiv\Q(\h)$, which is generated by $\1$ and the self--adjoint elements $\{R^{\pm}(\varphi)\}_{\varphi\in\h}$ so that
$$
R^{\#}(\varphi_{1})R^{\#}(\varphi_{2})+R^{\#}(\varphi_{2})R^{\#}(\varphi_{1})=2\inner{\varphi_{1},\varphi_{2}}\delta_{+,-}\1, \qquad \varphi_{1},\varphi_{2}\in\h,
$$
where $\#$ denotes either $R^{+}$ or $R^{-}$. Naturally, $\A$ and $\Q$ are isomorphic \cite{carlieb93}. Note that, view can algebras, $\A$ and $\Q$ are isomorphic to the Grassmann algebra $\wedge^{*}\h$ because they have exactly the same dimension: $\dim\wedge^{*}\h=2^{\dim\h}=\dim\A=\dim\Q$. Moreover, by using Definition \ref{definition star} below, we can endow to Grassmann algebras of a well--defined norm (see expression \ref{def norm}) such that the $\C$--algebras $\dim\wedge^{*}\h=2^{\dim\h}=\dim\A=\dim\Q$ are equivalent. The latter gain relevance for \emph{fermionic coherent states} because we are able to prove \emph{transparent} properties associated to the \emph{displacement} fermion operator, which are reminiscent to the bosonic version applied on quantum optics \cite{cahill99}. See Definition \ref{def:bil_weyl} and Lemma \ref{lemma:Weyl_grass} below.\\
Our main result is Theorem \ref{theo:main_result}, and the set of Lemmata \ref{lemma:Weyl_grass}, \ref{lemma:qfi}, \ref{lemma:Bruijn} and \ref{lemma:stam} are pivotal. Theorem \ref{theo:main_result} refers to the fermionic entropy power inequality for a \emph{beam--splitter} quantum channel operation. Its extension to \emph{amplifiers} quantum channels follows directly from the methods used in \cite{palmargio14}. We omit the mathematical details for the sake of simplicity. Additionally, we prove Lemmata \ref{lemma:qfi} and \ref{lemma:Bruijn} even for infinite dimensional $\CAR$ $\C$--algebras, which differ in their proofs of their parallel bosonic counterpart given in \cite{hubKonVer,datta17}.\\
To conclude, this paper is organized as follows:
\begin{itemize}
 \item In Section \ref{sec:mathe}, we present the mathematical framework of $\CAR$ $\C$--algebras. As a first step we present the self--dual $\CAR$ $\C$--algebras introduced by Araki \cite{A68,A70}, which were elegantly raised in the study of \emph{non}--interacting but \emph{non}--gauge fermion systems. Secondly, we provide the well--known Grassmann algebras $\G$, as well as the Berezin integrals and we define a ``circle'' product such that it converts the self--dual $\CAR$ algebras and the Grassmann algebras to be equivalent. After that, the Clifford $\C$--algebras, isomorphic to $\A$ and $\G$, are introduced. Then, a \emph{non}--commutative calculus is shown, which will provide us the differential equation determining the evolution of fermionic open systems. At the end of the section general properties of the set of states are exposed. 
 \item Section \ref{sec:main_results} provides the statement of the main result, Theorem \ref{theo:main_result}, and some main definitions concerning the entropy for the current work are defined. To be precise, we conveniently use the Clifford $\C$--algebra to define information quantities such as the quantum Fisher information and the entropy variation rate.   
 \item Section \ref{sec:tech_proofs} devotes to all technical proofs. Fermionic versions of well--known bosonic versions are proven. Additionally, as an application of our results, in the fermionic coherent states context, is the derivation of a natural mathematical framework only requiring the circle product presented in Section \ref{sec:mathe}.
 \item We finally include Appendix \ref{appen:qin}, stating a $\C$--algebra mathematical framework on the study of open quantum systems. This will permit to the  non--experts at quantum information theory and open quantum systems understand their mathematical motivations.
\end{itemize}
\section{Mathematical Framework}\label{sec:mathe}
\subsection{\eqt{$\CAR$ $\C$}--algebras}
\subsubsection{Self--dual \eqt{$\CAR$} algebras}
From now on and through of all the paper, let $\H$ be a finite--dimensional (complex) Hilbert space with even dimension $\dim\H\in 2\N$, and let $\a$  be an antiunitary involution on $\H$ such that $\a^{2}=\mathbf{1}_{\H}$ and 
\begin{equation*}
\left\langle \a\varphi _{1},\a\varphi _{2}\right\rangle
_{\H}=\left\langle \varphi _{2},\varphi _{1}\right\rangle _{\H}\ ,\qquad \varphi _{1},\varphi _{2}\in \H.
\end{equation*}
$\H$ endowed with $\a$, denoted $(\H,\a)$ is named a \emph{self--dual Hilbert space} and yields \emph{self--dual $\CAR$ algebra}, $\A\equiv(\mathrm{sCAR}(\H,\a),+,\cdot,^{*},\Vert\cdot\Vert_{\sCAR})$, which is nothing but a $C^{*}$--algebra generated by a unit $\1$ and a family $\{\B(\varphi )\}_{\varphi \in \H}$ of elements satisfying: $\B\left( \varphi \right) ^{*} $ is (complex) linear, $\B(\varphi )^{*}=\B(\a(\varphi ))$ for any $\varphi \in \H$ and the family $\{\B(\varphi )\}_{\varphi \in \H}$ satisfies the Canonical Anti--Commutation Relations ($\CAR$): For any $\varphi _{1},\varphi _{2}\in\H$,
\begin{equation}
\B(\varphi _{1})\B(\varphi _{2})^{*}+\B(\varphi_{2})^{*}\B(\varphi _{1})=\left\langle \varphi _{1},\varphi_{2}\right\rangle _{\H}\1.  \label{eq: CAR}
\end{equation}
Note that for any $\varphi\in\H$, $\Vert\B(\varphi)\Vert_{\A}\leq\Vert\varphi\Vert_{\H}$ with $\Vert\cdot\Vert_{\A}\equiv\Vert\cdot\Vert_{\sCAR}$. Additionally, $\A$ is isomorphic to the $\C$--algebra $\otimes^{\dim\H/2}\mathrm{Mat}(2,\CP)$, where for $N\in\N$, $\mathrm{Mat}(N,\CP)$ denotes the complex matrices of size $N\times N$ , see \eqref{dimension fock} and \cite{BratteliRobinson}.\\
For any self--dual Hilbert space $(\H,\a)$ we introduce:
\begin{definition}[Basis projections]\label{def basis projection}
A basis projection associated with $(\H,\a)$ is an
orthogonal projection $P\in \mathcal{B}(\H)$ satisfying $\a P\a=P^{\bot }\doteq \mathbf{1}_{\H}-P$. We denote by $
\mathfrak{h}_{P}$ the range $\mathrm{ran}P$ of the basis projection $P$. The set of all basis projections on $(\H,\a)$ it will denoted by $\fp$.
\end{definition}
For any $P\in\fp$, we can identify $\H$ with
\begin{equation}
\H\equiv \mathfrak{h}_{P}\oplus \mathfrak{h}_{P}^{*}
\label{definition H bar}
\end{equation}
and 
\begin{equation}
\B\left( \varphi \right) \equiv \B_{P}(\varphi )\doteq\B\left(P\varphi\right)+\B\left(\a P^{\bot }\varphi\right)^{*}.\label{map iodiote}
\end{equation}
Therefore, there is a natural isomorphism of $C^{*}$--algebras from $\A$ to the $\CAR$ algebra $\mathrm{CAR}(\mathfrak{h}_{P})$ generated by the unit $\1$ and $\{\B_{P}(\varphi )\}_{\varphi \in \mathfrak{h}_{P}}$. See Expression \eqref{eq: CAR basic}. In other words, a basis projection $P$ can be used to \emph{fix} so--called \emph{annihilation} and \emph{creations} operators.\\
For any unitary operator $U\in \mathcal{B}(\H)$ such that $U\mathfrak{A=A}U$, the family of elements $\B(U\varphi )_{\varphi \in\H}$, together with the unit $\1$, generates $\A$. In the latter case, $U$ is named a \emph{Bogoliubov transformation}, and the unique $^{*}$--automorphism $\mathbf{\chi }_{U}$ such that 
\begin{equation}\label{eq:bog_aut}
\mathbf{\chi }_{U}\left( \B(\varphi )\right) =\B(U\varphi ),\qquad \varphi \in \H,
\end{equation}
is called in this case a \emph{Bogoliubov }$^{*}$\emph{--automorphism}. Note that a Bogoliubov transformation $U\in \mathcal{B}(\H)$ always satisfies: $\det\left( U\right) =\pm 1$. If $\det\left( U\right) =1$, we say that $U$ has \emph{positive} orientation. Otherwise $U$ is said to have \emph{negative} orientation. These properties are also called even and odd.\\
Considering the Bogoliubov $^{*}$--automorphism \eqref{eq:bog_aut} with $U=-\mathbf{1}_{\H}$, an element $A\in\A$, satisfying
\begin{equation}\label{eq:even odd}
\mathbf{\chi}_{-\mathbf{1}_{\H}}(A)=
\begin{cases}
   \quad A&\text{is called \emph{even}},\\
   -A&\text{is called \emph{odd}}.
\end{cases}
\end{equation}
The subspace of even elements $\A^{+}$ is a sub--$C^{*}$--algebra of $\A$.\par
In order to study \emph{non}--interacting fermion systems, as is the case of the (reduced) BCS model at condensed matter physics, or Gaussian states at fermionic quantum computation, it is useful to introduce for $H\in \mathcal{B}(\H)$ its \emph{bilinear element} by
\begin{align}
\inner{\B,H\B}\doteq \sum\limits_{i,j\in I}\left\langle \psi _{i},H\psi _{j}\right\rangle _{\H}\B\left( \psi _{j}\right) \B\left( \psi _{i}\right) ^{*},\label{eq:bil_self}
\end{align}
where $\{\psi _{i}\}_{i\in I}$ is an orthonormal basis of $\H$. Note that $\inner{\B,H\B}$ is \emph{uniquely} defined in the sense that \emph{does not depend} on the particular choice of the orthonormal basis, but does depend on the choice of generators $\{\B(\varphi )\}_{\varphi \in \H}$ of the self--dual $\CAR$ algebra $\A$. Moreover, $\langle \B,H\B\rangle ^{*}=\langle\B,H^{*}\B\rangle$ for all $H\in \mathcal{B}(\H)$. The analysis of bilinear elements can be restricted to self--dual operators:
\begin{definition}[Self--dual operators]\label{def:self_dual}
A self--dual operator on $(\H,\a)$ is an operator $H\in\mathcal{B}(\H)$ satisfying the equality $H^{*}=-\a H\a$. If, additionally, $H$ is self--adjoint, then we say that it is a self--dual Hamiltonian on $(\H,\a)$.
\end{definition}
A basis projection $P$ (Definition \ref{def basis projection}) (block--) ``diagonalizes'' the self--dual operator $H\in \mathcal{B}(\H)$ whenever
\begin{equation}
H=\frac{1}{2}\left( PH_{P}P-P^{\bot }\a H_{P}^{*}\a P^{\bot }\right),\qquad \text{with}\qquad H_{P}\doteq2PHP\in\mathcal{B}(\mathfrak{h}_{P}).  \label{kappabisbiskappabisbis}
\end{equation}
In this situation, we also say that the basis projection $P$ diagonalizes $\inner{\B,H\B}$. On the other hand, given some self--dual Hamiltonian $H\in\BL(\H)$, and basis projection $P\in\fp$ with range $\h_{P}$, one can define a self--adjoint operator $H_{P}=H_{P}^{*}\in\BL(\h_{P})$ and the antilinear operator $G_{P}^{*}=-G_{P}$ on $\h_{P}$ as
\begin{equation}\label{eq:super_terms}
H_{P}\doteq2PHP \qquad\text{and}\qquad G_{P}\doteq2PH\a P.
\end{equation}
With this notation $\h_{P}$ is the so--called \emph{one--particle} Hilbert space while $H_{P}$ and $G_{P}$ are the gauge invariant and not gauge invariant \emph{one--particle} Hamiltonians respectively. In fact, for elements $\{\varphi_{i}\}_{i\in J}\in\h_{P}$, in the $\CAR$ $\C$--algebra setting (with generators $\1$ and $\{a(\varphi_{i})\}_{i\in J}$, see Expression \ref{eq: CAR basic}), one can write \emph{any} quadratic fermionic Hamiltonian as linear combinations of \emph{gauge--invariant} elements $a(\varphi_{i})a(\varphi_{j})^{*}$, for all $i,j\in J$, and linear combinations of \emph{non--gauge--invariant} elements $a(\varphi_{i})a(\varphi_{j})$ and $a(\varphi_{i})^{*}a(\varphi_{j})^{*}$, for all $i,j\in J$. Then, from \eqref{eq:bil_self} one note that any quadratic fermionic Hamiltonian can be recognized as
\begin{equation*}
\mathrm{d}\Gamma(H_{P})+\mathrm{d}\Upsilon(G_{P})=-\langle \mathrm{B},\left[\kappa \left(H_{P}\right) +\tilde{\kappa}\left(G_{P}\right)\right] \mathrm{B}\rangle +\frac{1}{2}\mathrm{Tr}_{\h_{P}}\left(H_{P}\right) \mathfrak{1},
\end{equation*}
with 
$$
\kappa\left(H_{P}\right) \doteq \frac{1}{2}\left(PH_{P}P-\a PH_{P}P\a\right)\in\BL(\H),\qquad\tilde{\kappa}\left(G_{P}\right) \doteq \frac{1}{2}\left( PG_{P}P\a-\a PG_{P}P\right) \in \mathcal{B%
}(\mathcal{H}),
$$
and $\mathrm{d}\Upsilon(G_{P})=-\inner{\B,\tilde{\kappa}\left(G_{P}\right)\B}$. Additionally, the $\CAR$ $\C$--algebra $\mathrm{CAR}(\h_{P})$ and the self--dual $\CAR$ $\C$--algebra $\mathrm{sCAR}(\H,\a)$ are the same $C^{\ast}$--algebra, by defining 
$$
\mathrm{B}\left( \varphi \right)\doteq a(\mathrm{\varphi }_{1})+a(\mathrm{\varphi }_{2})^{\ast },\qquad \varphi =(\mathrm{\varphi }_{1},\mathrm{\varphi }_{2}^{\ast })\in\h_{P}\oplus\h_{P}^{*}.
$$
\subsubsection{Grassmann Algebras}\label{Section Grassmann}
Consider the self--dual Hilbert space $(\H,\a)$. Grassmann algebras, also called exterior algebras are defined as follows: For every $n\in {\mathbb{N}}$ and $\varphi
_{1},\ldots ,\varphi _{n}\in \H$, we define the completely
antisymmetric $n$-linear form $\varphi _{1}\wedge \cdots \wedge \varphi _{n}$
from $\H^{n}$ to $\mathbb{C}$ by
\begin{equation*}
\varphi _{1}\wedge \cdots \wedge \varphi_{n}(\psi_{1},\ldots,\psi_{n})\doteq \det \left( (\varphi _{k}^{*}(\psi _{l}))_{k,l=1}^{n}\right)=\det\left((\left\langle \varphi _{k},\psi _{l}\right\rangle _{\H})_{k,l=1}^{n}\right),\qquad\psi_{1},\ldots ,\psi _{n}\in \H.
\end{equation*}
Then, using the definitions $\wedge ^{*0}\H\doteq {\mathbb{C}}$
and, for $n\in \mathbb{N}$, 
\begin{equation}
\wedge ^{*n}\H\doteq \mathrm{lin}\{\varphi _{1}\wedge \cdots
\wedge \varphi _{n}\colon\varphi _{1},\ldots ,\varphi _{n}\in \H^{*
}\equiv \H\},  \label{n-space grassmann}
\end{equation}%
we denote by 
\begin{equation}
\wedge ^{*}\H\doteq \bigoplus\limits_{n=0}^{\infty }\wedge^{*n}\H\equiv (\wedge ^{*}\H,+,\wedge )
\label{grassmann algebra}
\end{equation}
the Grassmann (associative and distributive) algebra associated with the
(self--dual) Hilbert space $(\H,\a)$. Here, the exterior product is defined, for any $n,m\in {\mathbb{N}}_{0}$, $\xi \in \wedge ^{*n}\H$ and $\zeta \in\wedge^{*m}\H$, by 
\begin{equation*}
\xi \wedge \zeta \left( \psi _{1},\ldots ,\psi _{n+m}\right) \doteq \frac{1}{n!m!}\sum_{\pi \in \mathcal{S}_{n+m}}\left( -1\right) ^{\pi }\xi \left( \psi
_{\pi (1)},\ldots ,\psi _{\pi (n)}\right) \zeta \left( \psi _{\pi(n+1)},\ldots ,\psi _{\pi (n+m)}\right),
\end{equation*}
where $\mathcal{S}_{N}$ is the set of all permutations of $N\in \mathbb{N}$ elements and $\psi _{1},\ldots ,\psi _{n+m}\in \H$. Obviously, for any $n\in {\mathbb{N}}$, $n\geq 2$, and $\varphi _{1},\ldots ,\varphi_{n}\in \H$,
\begin{equation}
\varphi _{1}\wedge \varphi _{2}=-\varphi _{2}\wedge \varphi _{1}\qquad \text{%
and}\qquad \varphi _{1}\wedge \left( \varphi _{2}\wedge \cdots \wedge
\varphi _{n}\right) =\varphi _{1}\wedge \cdots \wedge \varphi _{n}.
\label{grassmana anticommute00}
\end{equation}%
In the sequel, when there is no risk of ambiguity, we use 
\begin{equation}
\varphi_{1}\wedge\cdots\wedge\varphi_{n}\equiv\varphi _{1}\cdots\varphi_{n},\qquad \varphi _{1},\ldots,\varphi_{n}\in\H.
\label{eq:edge_product}
\end{equation}%
The unit of the Grassmann algebra $\wedge ^{*}\H$ is denoted by 
\begin{equation*}
\1\doteq 1\in \wedge ^{*0}\H\subseteq \wedge ^{*}\H
\end{equation*}
and $[\xi ]_{n}$ stands for the $n$-degree component of any element $\xi $
of $\wedge ^{*n}\H$, with $n\in \mathbb{N}_{0}$. Note also that 
\begin{equation*}
\H\equiv \H^{*}\doteq \wedge ^{*1}\H.
\end{equation*}
The subspace of $\wedge ^{*}\H$ generated by monomials $\varphi_{1}\cdots \varphi _{n}$ of even order $n\in 2{\mathbb{N}}_{0}$ forms a commutative subalgebra, the \emph{even subalgebra} of $\wedge ^{*}\H$, which is denoted by $\wedge _{+}^{*}\H$ in the sequel.\par
For any (complex) Hilbert space $\H$ and antiunitary involution $\a$ on $\H$ we can define a self--dual CAR algebra $\A$. The linear spaces $\A$ and $\wedge ^{*}\H$ are isomorphic to each other because they have exactly the same dimension:
\begin{equation}\label{dimension fock}
\dim\A=2^{\dim \H}=\dim \left(\wedge ^{*}\H\right).
\end{equation}
However, because of the CAR \eqref{eq: CAR} and the involution $\a$, $(\wedge ^{*}\H,+,\wedge)$ is not isomorphic to a self--dual $\CAR$ algebra over $\H$ and, following \cite{LD1}, we introduce the \emph{circle product} at Definition \ref{definition star} as well as an involution in order to make $\wedge ^{*}\H$ a self--dual CAR algebra.\par
For any $\varphi \in \H$, the linear operator $\delta/\delta \varphi$ acting on the Grassmann algebra $\wedge ^{*}\H$ is called \emph{Berezin derivative}, which is uniquely defined by the conditions 
\begin{equation}
\frac{\delta }{\delta \varphi }\tilde{\varphi}=\left\langle \varphi ,\tilde{%
\varphi}\right\rangle _{\H}\1\qquad \text{and}\qquad 
\frac{\delta }{\delta \varphi }\xi _{1}\xi _{2}=\left( \frac{\delta }{\delta
\varphi }\xi _{1}\right) \wedge \xi _{2}+\left( -1\right) ^{n}\xi _{1}\wedge
\left( \frac{\delta }{\delta \varphi }\xi _{2}\right),
\label{derivation plus1}
\end{equation}
for any $\tilde{\varphi}\in \H$ and element $\xi _{1}\in \wedge
^{*n}\H$ of degree $n\in \N$, and all $\xi _{2}\in
\wedge ^{*}\H$. \\
For each $k\in \N_{0}$, $\H^{(k)}$ denotes a copy of the Hilbert space $\H$ and the corresponding copy of $\xi\in\H$ is written as $\xi ^{(k)}$. For any $K\subset \{0,\ldots ,N\}$ with $N\in \N_{0}$, we identify $\wedge ^{*}(\oplus _{k\in K}\H^{(k)})$ with the Grassmann subalgebra of $\wedge ^{*}(\oplus _{k=0}^{N}\H^{(k)})$ generated by the union
\begin{equation*}
\bigcup\limits_{k\in K}\left\{\varphi^{(k)}\colon\varphi \in \H\right\}.
\end{equation*}%
We meanwhile identify $\wedge ^{*}\H^{(0)}$ with the Grassmann algebra $\wedge^{*}\H$, i.e.,
\begin{equation}
\wedge ^{*}\H^{(0)}\equiv \wedge ^{*}\H.
\label{identification}
\end{equation}
Taking into account this, we define:
\begin{definition}[Berezin integral]\label{Grassmann Integral}
Let $N\in \N_{0}$ and consider a basis projection $P$ (Definition \ref{def basis projection}) with $\{\psi _{i}\}_{i\in J}$ being any orthonormal basis of its range $\mathfrak{h}_{P}$. For all $k\in \{0,\ldots ,N\}$, we define the linear map 
\begin{equation*}
\int_{P}\d\left( \H^{(k)}\right) \colon\wedge ^{*}\left(
\oplus _{q=0}^{N}\H^{(q)}\right) \to \wedge ^{*}\left(
\oplus _{q\in \{0,\ldots ,N\}\backslash \{k\}}\H^{(q)}\right)
\end{equation*}
by 
\begin{align*}
\int_{P}\d\left( \H^{(k)}\right) \doteq \prod\limits_{i\in
J}\left( \frac{\delta }{\delta \psi _{i}^{(k)}}\frac{\delta }{\delta ((%
\a\psi _{i})^{(k)})}\right).
\end{align*}
\end{definition}
For $N=0$, the Berezin integral defines a linear form from $\wedge^{*}\H^{(0)}\equiv \wedge ^{*}\H$ to $\CP\1\equiv \CP$. One can show that for any basis projection $P$ diagonalizing a self--dual operator $H\in\BL(\H)$, we have
\begin{equation}
\int_{P}\d\left( \H\right)\e^{\inner{\H,H\H}}=\det\left( H_{P}\right),
\label{equation a la con}
\end{equation}
where for an orthonormal basis $\{\psi _{i}\}_{i\in I}$ of $\H$, the bilinear element $\inner{\H,H\H}$ on the Grassmann algebra $\wedge ^{*}\H$ is \emph{uniquely} given by
\begin{equation}
\inner{\H,H\H}\doteq \sum\limits_{i,j\in I}\inner{\psi _{i},H\psi _{j}}_{\H}\left(\a\psi _{j}\right) \wedge \psi _{i},\label{eq:bil_grass}
\end{equation}
c.f. \eqref{eq:bil_self}. In particular, $\det\left(H_{P}\right) $ only depends on $H$ and the orientation of $P$, which was defined around expression \eqref{eq:bog_aut}. Gaussian Berezin integrals are then defined as follows:
\begin{definition}[Gaussian Berezin integrals]\label{def:gauss_integral}
For $C\in \mathcal{B}(\H)$ an invertible self--dual operator, the Gaussian Berezin integral with covariance $C\in \mathcal{B}(\H)$ is the linear map $\int \mathrm{d\mu }_{C}\left( \H\right) $ from $\wedge ^{*}\H$ to $\CP\1$ defined by 
\begin{equation*}
\int \mathrm{d\mu }_{C}\left( \H\right) \xi \doteq \det\left( \frac{C_{P}}{2}\right) \int_{P}\d\left( \H\right)\e^{\frac{1}{2}\langle \H,C^{-1}\H\rangle }\wedge\xi,\quad \xi \in \wedge ^{*}\H,
\end{equation*}
where $P\in\fp$ is any basis projection diagonalizing $C$ (see \eqref{kappabisbiskappabisbis}).
\end{definition}
It can be proven the following \cite{LD1}:
\begin{proposition}[Gaussian Berezin integrals as Pfaffians]
\label{Gaussian integral properties}
Let $C\in \mathcal{B}(\H)$ be any invertible self--dual operator.
Then, $\int \mathrm{d\mu }_{C}\left( \H\right) \1=\1$ while, for all $N\in \N_{0}$ and $\varphi _{0},\ldots,\varphi _{2N}\in \H$, 
\begin{align*}
\int \mathrm{d\mu }_{C}\left( \H\right)\varphi_{0}\cdots \varphi_{2N}=0\quad\text{and}\quad \int\mathrm{d\mu}_{C}\left( \H\right)\varphi _{1}\cdots \varphi _{2N}=\mathrm{Pf}\left[ \left\langle \a\varphi _{k},C\varphi _{l}\right\rangle _{\H}\right] _{k,l=1}^{2N}\1. 
\end{align*}
\end{proposition}
Take a basis projection $P\in\fp$ with range $\mathfrak{h}_{P}$. For all $i,j,k,l\in \N_{0}$:
\begin{equation}
\varkappa _{(i,j)}^{(k,l)}\colon\wedge ^{*}(\mathfrak{h}_{P}^{(i)}\oplus 
\mathfrak{h}_{P}^{*(j)})\to \wedge ^{*}(\mathfrak{h}_{P}^{(k)}\oplus \mathfrak{h}_{P}^{*(l)})  \label{def kappa1}
\end{equation}
is the unique isomorphism of linear spaces such that $\varkappa_{(i,j)}^{(k,l)}(z\1)=z\1$ for $z\in {\CP}$ and, for any $m,n\in \N_{0}$ so that $m+n\geq 1$, and all $\varphi_{1},\ldots ,\varphi _{m+n}\in \mathfrak{h}_{P}$, 
\begin{equation}
\varkappa _{(i,j)}^{(k,l)}\left( (\a\varphi _{1})^{(i)}\cdots (%
\a\varphi _{m})^{(i)}\varphi _{m+1}^{(j)}\cdots \varphi
_{m+n}^{(j)}\right) =(\a\varphi _{1})^{(k)}\cdots (\a%
\varphi _{m})^{(k)}\varphi _{m+1}^{(l)}\cdots \varphi _{m+n}^{(l)}
\label{def kappa2}
\end{equation}
with $\varphi _{1}\wedge \varphi _{2}\equiv \varphi _{1}\varphi _{2}$.\\
We can equip \emph{any} Grassmann algebra with a $\C$--algebra structure. In order to proceed we take a basis projection $P\in\fp$ in such a way that we introduce the circle product $\circ _{P}$ as follows:
\begin{definition}[Circle products with respect to basis projections]\label{definition star}
Fix $P\in\fp$ with range $\mathfrak{h}_{P}$ and recall (\ref{definition H bar}), that is, $\H\equiv \mathfrak{h}_{P}\oplus \mathfrak{h}_{P}^{*}$. For any $\xi _{0},\xi _{1}\in \wedge^{*}\H$, we define their circle product by
\begin{align*}
\xi _{0}\circ _{P}\xi _{1}\doteq \left( -1\right) ^{\frac{\dim\H}{2}}\int_{P}\d\left( \H^{(1)}\right) \varkappa_{(0,0)}^{(0,1)}(\xi _{0})\varkappa _{(0,0)}^{(1,0)}(\xi _{1})\e^{-\langle \mathfrak{h}_{P}^{(0)},\mathfrak{h}_{P}^{(0)}\rangle }\e^{\langle \mathfrak{h}_{P}^{(0)},\mathfrak{h}_{P}^{(1)}\rangle }\e
^{-\langle \mathfrak{h}_{P}^{(1)},\mathfrak{h}_{P}^{(1)}\rangle }\e
^{\langle \mathfrak{h}_{P}^{(1)},\mathfrak{h}_{P}^{(0)}\rangle }.
\end{align*}
\end{definition}
The space $(\wedge ^{*}\H,+)$ endowed with the circle product $\circ _{P}$ is an (associative and distributive) algebra, like $(\wedge ^{*}\H,+,\wedge)$, for \emph{any} $P\in\fp$. Among other properties of $\circ_{P}$ note that this satisfy the Canonical Anti-commutation Relations ($\CAR$): 
\begin{equation}\label{eq:CAR_gras}
 \varphi_{1}^{*}\circ_{P}\varphi_{2}+\varphi_{2}\circ_{P}\varphi_{1}^{*}=\inner{\varphi_{1},\varphi_{2}}_{\H}\1,\qquad\varphi_{1},\varphi_{2}\in\H.
\end{equation}
Furthermore, for any $P\in\fp$, and $\varphi _{1},\varphi_{2}\in \mathfrak{h}_{P}$, 
\begin{equation}\label{eq:prod_P}
\left( \a\varphi _{1}\right) \circ _{P}\varphi _{2}=\left( \a\varphi _{1}\right) \wedge\varphi _{2}.
\end{equation} 
We can endow $\wedge ^{*}\H$ with an involution, which turns $(\wedge ^{*}\H,+,\circ _{P})$ into a $^{*}$--algebra. Namely, one define the involution to satisfy for any $P\in\fp$ that $\1^{*}=\1$ and
\begin{align*}
(\varphi_{1}\circ_{P}\varphi_{2})^{*}=\left(\a\varphi_{2}\right)\circ_{P}\left(\a\varphi_{1}\right),\qquad n\in\N,\varphi_{1},\varphi_{2}\in\H.
\end{align*}
Hence, $(\wedge ^{*}\H,+,\wedge )$ equipped with the involution $^{*}$ is a $^{*}$--algebra, i.e.,
\begin{equation}\label{eq:star_algebra}
\left( \xi _{0}\wedge \xi _{1}\right) ^{*}=\xi _{1}^{*}\wedge \xi_{0}^{*},\qquad \xi _{0},\xi _{1}\in \wedge ^{*}\H.  
\end{equation}
Thus for a self--dual Hilbert space $(\H,\a)$ and $P\in\fp$, $(\wedge ^{*}\H,+,\circ_{P},^{*})$ is a $^{*}$--algebra generated by $\1$ and the family $\{\varphi ^{*}\}_{\varphi \in \H}$ of elements satisfying the same properties of the self--dual $\CAR$ algebras (see Expression \eqref{eq: CAR}), namely, \eqref{eq:CAR_gras}. Additionally, there is a canonical $^{*}$--isomorphism between a self--dual $\CAR$ algebra constructed from $\left(\H,\a\right) $ and $\wedge ^{*}\H$:
\begin{definition}[Canonical isomorphism of $^{*}$--algebra]
\label{def:isomor}
For $P\in\fp$, we define the canonical isomorphism 
\begin{equation*}
\varkappa_{P}\colon(\A,+,\cdot,^{*})\to(\wedge^{*}\H,+,\circ_{P},^{*})
\end{equation*}
via the conditions $\varkappa_{P}(z\1)=z\1$ and $\varkappa_{P}\left(\B(\varphi)\right)=\varphi^{*}$ for all $\varphi\in\H$. 
\end{definition}
Therefore note that bilinear elements of self--dual $\CAR$ algebra, see Equation \eqref{eq:bil_self}, are mapped via $\varkappa_{P}$ (up to some constant) to bilinear elements of Grassmann algebra, as stated in Definition \ref{def:bil_weyl}. In fact, one can proof that for any $P\in\fp$,
\begin{equation}
\varkappa_{P}\left(\langle\B,H\B\rangle\right)=\langle\H,H\H\rangle+\mathrm{Tr}_{\H}\left(P^{\bot}HP^{\bot}\right)\1,\qquad H\in \mathcal{B}(\H).\label{petit calcul2}
\end{equation}
For $P\in\fp$, we endow $(\wedge ^{*}\H,+,\circ _{P},^{*})$ with the norm
\begin{equation}
\left\Vert \xi \right\Vert_{\wedge^{*}\H}\doteq\left\Vert \varkappa _{P}^{-1}\left( \xi \right) \right\Vert _{\A},\qquad\xi\in\wedge^{*}\H,  \label{def norm}
\end{equation}
in order to do it a self--dual CAR ($C^{*}$--) algebra. In this case, $\varkappa _{P}$ is an isometry. In the sequel we will write $\mathcal{G}_{P}\equiv(\wedge^{*}\H,+,\circ_{P},^{*},\left\Vert \cdot\right\Vert_{\wedge^{*}\H})$, to lighten notation. Similarly, for the commutative even subalgebra $\wedge _{+}^{*}\H$ associated to $\wedge ^{*}\H$ (see \eqref{grassmana anticommute00} and comments around it), from now on, $\G_{P}^{+}$ will be a commutative $\C$--algebra.
\subsubsection{Clifford Algebras}\label{subsec:Clifford}
\emph{Clifford $\C$--algebras} are presented in the following way: For any basis projection $P\in\fp$, and any orthonormal basis $\{\psi_{j}\}_{j\in J}$ of range $\h_{P}$, we define the self--adjoint elements 
$$
Q_{j}\equiv Q(\psi_{j})\doteq\B(\psi_{j})^{*}+\B(\psi_{j})\quad\text{and}\quad P_{j}\equiv P(\psi_{j})\doteq\ii(\B(\psi_{j})^{*}-\B(\psi_{j})),
$$
for any $j\in J$. The family of elements $\{Q_{j}\}_{j\in J}$ (resp. $\{P_{j}\}_{j\in J}$) is known as the \emph{configuration operators} (resp. \emph{conjugate momenta
operators}) \cite{carlieb93}. At the fermionic information context, $|J|=\dim\H/2$ is the number of \emph{fermionic modes} of the physical system, while the operators  $\{Q_{j}\}_{j\in J}$ and $\{P_{j}\}_{j\in J}$ are known as the \emph{Majorana} fermion operators. Let 
\begin{equation}
\J\doteq J\times\{+,-\}.\label{eq:finite_setJ}
\end{equation}
We denote by $R_{j,+}=Q_{j}$ and $R_{j,-}=P_{j}$, for any $j\in J$. Note that the unit $\1$ and the family of self--adjoint elements $\{R_{\j}\}_{\j\in\J}$ generate the Clifford algebra $\Q\equiv(\Q,+,\cdot,^{*},\Vert\cdot\Vert_{\Q})$ of size $\dim\Q=2^{\dim\H}$, satisfying the $\CAR$:
\begin{equation}
R_{\mathfrak{i}}R_{\j}+R_{\j}R_{\mathfrak{i}}=2\delta_{\mathfrak{i},\j}\1\quad\text{with}\quad\delta_{\mathfrak{i},\j}\doteq\delta_{i,j}\delta_{s,t}.\label{eq:CAR_clif}
 \end{equation}
for $s,t\in\{+,-\}$. Note that $\C$--algebras $\A$ (self--dual $\CAR$ algebra) and $\Q$ have exactly the same dimension, and then are isomorphic, thus by \eqref{dimension fock} and Definition \ref{def:isomor}, $\Q$ is also isomorphic to the $\C$--algebra $\mathcal{G}$. Furhermore, similar to the algebras $\A$ and $\wedge^{*}\H$ cases, for any $H\in\BL(\H)$ one can \emph{uniquely} introduce the bilinear element $\inner{\mathrm{R},H\mathrm{R}}$ on $\Q$ by
\begin{equation}
\inner{\mathrm{R},H\mathrm{R}}\doteq\sum\limits_{\mathfrak{i},\j\in\J}\inner{\psi _{\mathfrak{i}},H\psi _{\j}}_{\H}R(\psi_{\j})R(\psi_{\mathfrak{i}}),\label{eq:bil_cliff}
\end{equation}
c.f. \eqref{eq:bil_self} and \eqref{eq:bil_grass}, where $\{\psi _{\mathfrak{i}}\}_{\mathfrak{i}\in\J}$ is an orthonormal basis of $\H$, with $\J$ given by \eqref{eq:finite_setJ}.\par
We can endow $\Q$ with the Hilbert--Schmidt inner product $\inner{\cdot,\cdot}_{\Q}^{\text{H.S.}}$ given by
\begin{equation}
\inner{A,B}_{\Q}^{\text{H.S.}}\doteq\tr_{\Q}(A^{*}B),\qquad A,B\in\Q,\label{eq:HSclifford} 
\end{equation}
where $\tr_{\Q}$ is the tracial state on $\Q$. Consider now the Bogoliubov $^{*}$--automorphism $\chi_{-\mathbf{1}_{\H}}\colon\A\to\A$ given by \eqref{eq:even odd} (and hence $\chi_{-\mathbf{1}_{\H}}\colon\Q\to\Q$ too), for $\j\in\J$, in order to introduce the \emph{skew--derivation} $\nabla_{\j}\colon\Q\to\Q$ as\footnote{See \cite{carlen2020non} for a general study of \emph{non--commutative calculus} at the framework of $\C$--algebras.}
$$
\nabla_{\j}(A)\doteq\frac{1}{2}\left(R_{\j}A-\chi_{-\mathbf{1}_{\H}}(A)R_{\j}\right),\qquad A\in\Q.
$$
Here, by skew we mean that for each $A,B\in\Q$ we have and \emph{anti--derivation Leibniz's law} property 
$$
\nabla_{\j}(AB)=\nabla_{\j}(A)B+\chi_{-\mathbf{1}_{\H}}(A)\nabla_{\j}(B).
$$
Additionally, since $\chi_{-\mathbf{1}_{\H}}(R_{\j})=-R_{\j}$ for any $\j\in\J$ we can consider the inner product $\inner{\cdot,\cdot}_{\Q}^{\text{H.S.}}$, such that $\inner{\nabla_{\j}(A),B}_{\Q}^{\text{H.S.}}=\inner{A,\nabla_{\j}^{*}(B)}_{\Q}^{\text{H.S.}}$, and $\inner{\chi_{-\mathbf{1}_{\H}}(A),BR_{\j}}_{\Q}^{\text{H.S.}}=-\inner{A,\chi_{-\mathbf{1}_{\H}}(B)R_{\j}}_{\Q}^{\text{H.S.}}$ in order to get
$$
\nabla_{\j}^{*}(A)=\frac{1}{2}\left(R_{\j}A+\chi_{-\mathbf{1}_{\H}}(A)R_{\j}\right),\qquad A\in\Q.
$$
Then, for any $A,B\in\Q$, one introduce the \emph{fermionic number operator} $\Nc$ on $\Q$ satisfying,
\begin{equation}
\mathcal{F}(A,B)\doteq\inner{A,\Nc B}_{\Q}^{\text{H.S.}},\label{eq:number_gross} 
\end{equation}
where $\mathcal{F}$ is the \emph{Gross's Fermionic Dirichlet} form $\mathcal{F}(A,A)$ on $\Q$ so that \cite{carlen14}
$$
\mathcal{F}(A,A)\doteq\tr_{\Q}\left(\left(\nabla A\right)^{*}\cdot\nabla A\right)=\sum_{\j\in\J}\tr_{\Q}(\left(\nabla_{\j}(A)\right)^{*}\cdot\nabla_{\j}(A)).
$$
Comparing the two latter equations, one note that it is possible to write
\begin{align*}
\mathcal{F}(A,A)&= \sum_{\j\in\J}\tr_{\Q}\left(\left(\frac{1}{2}\left(R_{\j}A-\chi_{-\mathbf{1}_{\H}}(A)R_{\j}\right)\right)^{*}\cdot\left(\frac{1}{2}\left(R_{\j}A-\chi_{-\mathbf{1}_{\H}}(A)R_{\j}\right)\right)\right)\\
&=\frac{1}{4}\sum_{\j\in\J}\tr_{\Q}\left(A^{*}\left(A-R_{\j}\chi_{-\mathbf{1}_{\H}}(A)R_{\j}\right)+\chi_{-\mathbf{1}_{\H}}(A)^{*}\left(\chi_{-\mathbf{1}_{\H}}(A)-R_{\j}AR_{\j}\right)\right),
\end{align*}
and using that $\inner{\chi_{-\mathbf{1}_{\H}}(A),A}_{\Q}^{\text{H.S.}}=\inner{A,\chi_{-\mathbf{1}_{\H}}(A)}_{\Q}^{\text{H.S.}}$ we obtain from \eqref{eq:number_gross} that
\begin{equation}
\Nc A=\frac{1}{2}\sum_{\j\in\J}\tr_{\Q}\left(A-R_{\j}\chi_{-\mathbf{1}_{\H}}(A)R_{\j}\right),\label{eq:number_gross2} 
\end{equation}
such that we finally introduce the ``Fermionic Mehler semigroup'' as $\{\P_{t}\}_{t\in\R_{0}^{+}}\doteq\left\{\e^{-t\Nc}\right\}_{t\in\R_{0}^{+}}$. This semigroup satisfies the differential equation
$$
\frac{\d }{\d t}\rho_{t}=-\Nc\rho_{t},\qquad\rho_{t}\doteq\e^{-t\Nc}\rho,
$$
for any density matrix $\rho\in\Q^{+}\cap\Q$ and $\rho_{0}\doteq\rho$.
\subsection{States on \eqt{$\CAR$ $\C$}--algebras}
Take $N\in\N$. If $M\in\mathrm{Mat}\left(2N,\CP\right)$ is a complex matrix of size $2N\times2N$ and satisfies $M_{k,l}=-M_{l,k}$ is said to be ``skew--symmetric'' or ``anti--symmetric''. If additionally $M$ is a normal matrix ($MM^{*}=M^{*}M$) then exists a unitary $U\in\mathrm{Mat}(2N,\CP)$ with $U^{\mathrm{t}}$ denoting its transpose such that $M=U^{\mathrm{t}}\Lambda U$, where $\Lambda$ is a block diagonal matrix of size $2N\times2N$ such that it can be decomposed as a direct sum of $N$ skew--symmetric matrices of size $2\times2$. More precisely 
\begin{equation}
\Lambda\doteq\bigoplus_{j=1}^{n}\Lambda_{j}\equiv\mathrm{diag}\left\{\Lambda_{1},\ldots,\Lambda_{N}\right\},\label{eq:skewLambda}
\end{equation}
where, for $j\in\{1,\ldots,N\}$, $\Lambda_{j}$ is a skew--symmetric matrix with entries $\left\{\Lambda_{j}\right\}_{12}=-\left\{\Lambda_{j}\right\}_{21}=\lambda_{j}\in\R$. Note that for $(\H,\a)$ a self--dual Hilbert space, a self--dual operator $C\in\BL(\H)$ and $N\in\N$, the complex matrix defined by 
\begin{equation}
C_{k,l}\doteq\inner{\a\varphi _{k},C\varphi _{l}}_{\H}, \qquad k,l\in\{1,\ldots,2N\}\label{eq:skewC} 
\end{equation}
is skew--symmetric.\par
We again consider the self--dual Hilbert space $(\H,\a)$, and consider the set of states of $\A$ (the self--dual $\CAR$ $\C$--algebra associated to $(\H,\a)$), denoted by $\states_{\A}\subset\A^{*}$\footnote{See Appendix \ref{appen:qin} for a general discussion of states on the $\C$--algebra setting.}. An important class of states are the so--called \emph{quasi--free} states, that are defined for all $N\in\N_{0}$ and $\varphi_{0},\ldots,\varphi_{2N}\in\H$ as
\begin{equation}
\omega\left(\B\left(\varphi_{0}\right)\cdots \B\left(\varphi_{2N}\right) \right)=0,  \label{ass O0-00}
\end{equation}
while, for all $N\in \N$ and $\varphi _{1},\ldots,\varphi _{2N}\in\H$,
\begin{equation}
\omega \left( \B\left( \varphi _{1}\right) \cdots \B\left(\varphi _{2N}\right) \right) =\mathrm{Pf}\left[\omega \left( \mathbb{O}_{k,l}\left( \B(\varphi _{k}),\B(\varphi _{l})\right)\right) \right] _{k,l=1}^{2N},  \label{ass O0-00bis}
\end{equation}
where $\mathrm{Pf}$ is Pfaffian of the $2N\times 2N$ skew--symmetric matrix $M\in \mathrm{Mat}\left( 2N,\CP\right)$ defined by 
\begin{equation}
\mathrm{Pf}\left[M_{k,l}\right] _{k,l=1}^{2N}\doteq \frac{1}{2^{N}N!}\sum_{\pi \in \mathcal{S}_{2N}}\left( -1\right) ^{\pi}\prod\limits_{j=1}^{N}M_{\pi \left(2j-1\right) ,\pi \left( 2j\right) }\label{Pfaffian}
\end{equation} 
and $\mathbb{O}_{k,l}$ by
\begin{equation*}
\mathbb{O}_{k,l}\left(A_{1},A_{2}\right)\doteq\left\{ 
\begin{array}{ccc}
A_{1}A_{2} & \text{for} & k<l, \\ 
-A_{2}A_{1} & \text{for} & k>l, \\ 
0 & \text{for} & k=l.
\end{array}%
\right.
\end{equation*}
Quasi--free states are therefore particular states that are uniquely defined by two-point correlation functions, via \eqref{ass O0-00}--\eqref{ass O0-00bis}. In fact, a quasi--free state $\omega\in\states_{\A}$ is uniquely defined by its so--called \emph{symbol}, that is, a positive operator $S_{\omega}\in \mathcal{B}(\H)$ such that%
\begin{equation}
0\leq S_{\omega }\leq \mathbf{1}_{\H}\qquad \text{and}\qquad S_{\omega
}+\a S_{\omega }\a=\mathbf{1}_{\H},
\label{symbol}
\end{equation}
through the conditions
\begin{equation}
\left\langle\varphi_{1},S_{\omega }\varphi _{2}\right\rangle _{\H}=\omega \left( \B(\varphi_{1})\B(\a\varphi_{2})\right),\qquad \varphi _{1},\varphi _{2}\in \H.\label{symbolbis}
\end{equation}
For more details on symbols of quasi--free states, see \cite[Lemma 3.2]{A70}. Conversely, any self--adjoint operator satisfying (\ref{symbol}) uniquely defines a quasi--free state through Equation (\ref{symbolbis}). In physics, $S_{\omega} $ is called the \emph{one--particle density matrix} of the system. An example of a quasi--free state is provided by the tracial state (cf. Expression \eqref{eq:tra_state}):
\begin{definition}[Tracial state]\label{def trace state}
The tracial state $\tr_{\A}\in\states_{\A}$ is the quasi--free state with symbol $S_{\tr}\doteq\frac{1}{2}\mathbf{1}_{\H}$.
\end{definition}
An important density matrix $\rho_{\omega}^{(\beta)}$ is that related to \emph{thermal} equilibrium states, or Gibbs states $\omega^{(\beta)}\in\states$ where $\beta\in\R^{+}$ is the inverse temperature. In this case, given any self--dual Hamiltonian $H$ on $(\H,\a)$ (Definition \ref{def:self_dual}), the positive operator 
$$
S_{H}^{(\beta)}\doteq\frac{1}{1+\e^{-\beta H}}
$$ 
satisfies Condition (\ref{symbol}) and for any $A\in\A$ is the symbol of a quasi--free state $\omega_{H}$ satisfying
\begin{equation}
\omega _{H}^{(\beta)}(A)=\frac{\tr_{\A}\left(A\e^{\frac{\beta }{2}\langle 
\B,H\B\rangle}\right) }{\tr_{\A}\left(\e^{\frac{\beta }{2}\langle 
\B,H\B\rangle}\right)}.
\label{Gibbs states}
\end{equation}
One verify that the self--dual Hamiltonian $H$ on $(\H,\a)$ give rise to the density matrix $\rho^{(\beta)}\in\A$
$$
\rho^{(\beta)}\doteq\frac{\e^{\frac{\beta }{2}\langle\B,H\B\rangle}}{\tr_{\A}\left(\e^{\frac{\beta }{2}\langle\B,H\B\rangle}\right)}.
$$
Physically, $\rho^{(\beta)}$ minimizes the \emph{free energy} of the physical system provided $H$. See \cite{BratteliRobinson} for details.\\
As already discussed, for any even size Hilbert space $\H$ with associated self--dual Hilbert space $(\H,\a)$, the algebras $\A,\wedge^{*}\H$ and $\Q$ algebras are isomorphic. More generally, for any basis projection $P\in\fp$ one can endow with an involution and a norm to $\wedge^{*}\H$ in such a way that $\A,\G_{P}$ and $\Q$ are $\C$--algebras. See Equations \eqref{petit calcul2}--\eqref{def norm} for notations. For the sake of simplicity, let $\cA$ to be $\A,\G_{P}$ or $\Q$. By Definition \ref{def:gauss_states_Calg}, for any \emph{invertible} bounded operator $C\in\BL(\H)$ providing a bilinear element $\mathrm{C}$, on $\cA$ (see \eqref{eq:bil_self}, \eqref{eq:bil_grass} and \eqref{eq:bil_cliff}) one can define a Gaussian state $\omega_{C}\in\states_{\cA}$ with associated density matrix $\rho_{C}\in\cA^{+}\cap\cA$, explicitly written as $\rho_{C,\cA}\doteq\frac{\e^{\alpha \mathrm{C}}}{\tr_{\cA}\left(\e^{\alpha \mathrm{C}}\right)}$ with $\alpha\in\CP$. Similarly, for $M\in\R^{+}$, the operator $g_{A,\cA}=M\e^{\alpha A}$ is called gaussian, not necessarily normalized. $C$ is called the covariance of the density matrix $\rho_{C}$. Note that by the isomorphism $\varkappa_{P}$ of Definition \ref{def:isomor} one can obtain similar fermion representations at the algebras $\A$ and $\G_{P}$. Additionally, note that for any $P\in\fp$ and any invertible operator $C\in\BL(\H)$, the isomorphism $\varkappa_{P}$ of Definition \ref{def:isomor} relates Gaussian operators $g_{C,\A}\in\A,g_{C,\G_{P}}\in\G_{P}$ and $g_{C,\Q}\in\Q$ by
$$
B_{C,P}\varkappa_{P}\left(g_{C,\A}\right)=E_{C,P}g_{C,\G_{P}}=D_{C,P}\varkappa_{P}\left(g_{C,\Q}\right),
$$
explicitly $B_{C,P}\e^{\varkappa_{P}\left(\inner{\B,C^{-1}\B}\right)}=E_{C,P}\e^{\inner{\H,C^{-1}\H}}=D_{C,P}\e^{\varkappa_{P}\left(\inner{\mathrm{R},C^{-1}\mathrm{R}}\right)}$, with $B_{C,P},C_{C,P},D_{C,P}\in\R^{+}$ positive numbers depending on $C$ and $P$. See Definition \ref{def:gauss_integral}. In particular observe that the Gibbs state $\omega_{H}^{(\beta)}\in\states_{\A}$ given by \eqref{Gibbs states} is Gaussian. One can inquire about the relation between $\det(\cdot)$ and $\tr(\cdot)$ while comparing the positive numbers $B_{C,P},E_{C,P}, C_{C,P}$. See again Definition \ref{def:gauss_integral} and note that our definitions coincide with those given in \cite{DeNaSol}. In the scope of a general setting, for any \emph{covariance} matrix $C_{H,P}^{(\beta)}$ depending on $H\in\BL(\H),P\in\fp$ and $\beta\in\R^{+}$ (see \cite[Corollary 4.8]{LD1} for a concrete Definition) we can write the Determinant of $C_{H,P}^{(\beta)}$ as a trace of well--defined product of operators defined via $H\in\BL(\H),P\in\fp$ and $\beta\in\R^{+}$, see \cite[Theorem 5.1]{LD1}.
\section{Main results}\label{sec:main_results}
Let $\W$ be a finite unital $\C$--algebra. As is usual, for the tracial state $\tr_{\W}\in\states_{\W}$ and any state $\omega\in\states_{\W}$ with associated density matrix $\rho_{\omega}\in\W^{+}\cap\W$, the von Neumann entropy $\mS\colon\states_{\W}\to\R$ is given by
\begin{equation}
\mS(\omega)\doteq-\tr_{\W}(\rho_{\omega}\ln\rho_{\omega}).\label{eq:vNentro} 
\end{equation}
Similarly, for the states $\omega_{1},\omega_{2}\in\states_{\W}$, 
\begin{equation}\label{eq:rel_entro}
\S(\omega_{1}\Vert \omega_{2})\doteq
\begin{cases}
\tr_{\W}\left(\rho_{\omega_{1}}\left(\ln\rho_{\omega_{1}}-\ln\rho_{\omega_{2}}\right)\right),&\text{if}\quad\mathrm{supp}(\rho_{\omega_{2}})\geq\mathrm{supp}(\rho_{\omega_{1}})
,\\
+\infty,&\text{otherwise},
\end{cases}
\end{equation}
denotes the entropy of $\omega_{1}$ relative to $\omega_{2}$. In \eqref{eq:rel_entro}, for the state $\omega\in\states_{\W}$, $\mathrm{supp}(\omega)$ denotes its support defined by the smallest projection $P\in\W$ such that $\omega(P)=1$. The \emph{quantum entropy power} associated to the state $\omega\in\states_{\W}$ is defined by
\begin{equation}
\E(\omega)\doteq\e^{\frac{\mS(\omega)}{\Nc_{\W}}}\in\R^{+},\label{eq:qep}
\end{equation}
where $\Nc_{\W}\doteq|\W^{+}\cap\W|\in\N$ is the number of \emph{modes} of the physical system described via $\W$. \\
Consider now the Clifford $\C$--algebra $\Q$ given at Subsection \ref{subsec:Clifford}. For the state $\omega\in\states_{\Q}$, with associated density matrix $\rho\equiv\rho_{\omega}\in\Q^{+}\cap\Q$, and the family of self--adjoint elements $\{R_{\j}\}_{\j\in\J}$ of $\Q$ with $\J$ given by \eqref{eq:finite_setJ}, the ``quantum Fisher information'' is
\begin{equation}
\mJ\left(\omega_{R_{\j}}\right)\doteq\frac{\d^{2}}{\d\theta^{2}}\S\left(\omega\Vert\omega_{R_{\j}}^{(\theta)}\right)\Big|_{\theta=0},\label{eq:qfi}
\end{equation}
where $\omega_{R_{\j}}^{(\theta)}$ is such that its associated density matrix is given by
\begin{equation}\label{eq:displa_clif2}
\rho_{R_{\j}}^{(\theta)}\doteq\e^{\theta R_{\j}}\rho\e^{-\theta R_{\j}}.
\end{equation}
Thus the ``entropy variation rate'' is defined by \cite{datta17} (see also \cite{hubKonVer})
\begin{equation}
\mJ\left(\omega\right)\doteq\sum_{\j\in\J}\mJ\left(\omega_{R_{\j}}\right).\label{eq:entropy_var_rate} 
\end{equation}\par
Let now $\mathbf{A}$ and $\mathbf{B}$ be two interacting fermion systems such that the product $\C$--algebra $\cA_{\mathbf{I}}\equiv\cA_{\mathbf{A}}\otimes\cA_{\mathbf{B}}$ describes the interacting system, where $\cA_{\mathbf{A}}$ and $\cA_{\mathbf{B}}$ are the $\CAR$ $\C$--algebras describing $\mathbf{A}$ and $\mathbf{B}$, respectively. For Gaussian states $\omega_{\textbf{A}}\in\states_{\textbf{A}}\equiv\states_{\cA_{\textbf{A}}}\subset\cA_{\textbf{A}}'$ and $\omega_{\textbf{B}}\equiv\states_{\cA_{\textbf{B}}}\in\cA_{\textbf{B}}'$, the tensor product of density matrices $\rho_{\textbf{A}}\otimes\rho_{\textbf{B}}$ equals the density matrix $\rho_{\mathbf{I}}$. Moreover, by taking into account Appendix \ref{appen:qin}, the unitary bounded operator $\mathbb{U}\in\BL(\cA_{\mathbf{I}})$ on $\cA_{\mathbf{I}}$ satisfying $\mathbb{U}\mathbb{U}^{*}=\1$ defines a \emph{quantum channel} given by
\begin{equation}
\E_{\mathbb{U}}(\omega_{\mathbf{I}})\doteq\tr_{\mathbf{B}}\left(\mathbb{U}^{*}\left(\rho_{\mathbf{A}}\otimes\rho_{\mathbf{B}}\right)\mathbb{U}\right),\label{eq:quant_channe}
\end{equation}
with $\omega_{\textbf{I}}\equiv\omega_{\textbf{A}}\otimes\omega_{\textbf{B}}$. Then a set of results regarding entropy variation rate of $\omega_{\textbf{A}},\omega_{\textbf{B}}$ and $\omega_{\textbf{I}}$ are summarized in the following Corollary:
\begin{corollary}\label{corollary:Gaussian}
Let $\mathbf{A}$ and $\mathbf{B}$ be two fermion systems and take same assumptions of Lemma \ref{lemma:qfi} for the Gaussian states $\omega_{\mathbf{A}}$ and $\omega_{\mathbf{B}}$. Therefore:
\begin{enumerate}
 \item Take $\alpha\in\R$, $\mJ(\omega_{\mathbf{C},\alpha R_{\j}})=\alpha^{2}\mJ(\omega_{\mathbf{C},R_{\j}})$ for any $R_{\j}$ generator of the Clifford $\C$--algebra $\Q_{\mathbf{C}}$, for $\mathbf{C}\in\{\mathbf{A},\mathbf{B}\}$.
 \item $\mJ(\omega_{\mathbf{A}}),\mJ(\omega_{\mathbf{B}})\in\R_{0}^{+}$.
 \item $\mJ(\omega_{\mathbf{A}}\otimes\omega_{\mathbf{B}})=\mJ(\omega_{\mathbf{A}})+\mJ(\omega_{\mathbf{B}})$.
\item Let $\mathbb{U}\in\BL(\cA_{\mathbf{I}})$ be a unitary bounded operator defining a quantum channel $\E_{\mathbb{U}}$. Then, for $\mathbf{C}\in\{\mathbf{A},\mathbf{B},\mathbf{I}\}$ we have 
\begin{align*}
\mJ\left(\omega_{\mathbf{C}}\right)\geq\mJ\left(\E_{\mathbb{U}}\left(\omega_{\mathbf{C}}\right)\right).
\end{align*}
\end{enumerate}
\end{corollary}
\begin{proof}
In order to prove the assertions we remind some pivotal properties of the relative entropy. See \cite{lindblad} and \cite{Wehrl}. Consider the product $\CAR$ $\C$--algebra $\cA_{\mathbf{I}}\equiv\cA_{\mathbf{A}}\otimes\cA_{\mathbf{B}}$ and the states $\omega_{1},\omega_{2}\in\states_{\cA_{\mathbf{I}}}$. We have (i) \emph{non--negativity}: $\S(\omega_{1}\Vert \omega_{2})\geq0$, (ii) \emph{monotonicity}: $\S(\omega_{1}\Vert \omega_{2})\geq\S(\E_{\mathbb{U}}(\omega_{1})\Vert\E_{\mathbb{U}}(\omega_{2}))$ for any the quantum channel given by \eqref{eq:quant_channe} and (iii) \emph{additivity}: $\S(\omega_{1}\otimes\omega_{2}\Vert\sigma_{1}\otimes\sigma_{2})=\S(\omega_{1}\Vert\sigma_{1})+\S(\omega_{2}\Vert\sigma_{2})$ for any normal states $\omega_{1},\sigma_{1}\in\states_{\cA_{\textbf{A}}}$ and $\omega_{2},\sigma_{2}\in\states_{\cA_{\textbf{B}}}$. Thus Part 1. follows from Lemma \ref{lemma:qfi}, while parts 2., 3 and 4 can be shown in a similar way that is done in \cite{konSmi14}, where the authors take into account (i), (ii) and (iii) properties. 
\end{proof}
The fermionic version of the quantum entropy power inequality is stated as follows:
\begin{theorem}[Fermionic Entropy Power Inequality]\label{theo:main_result}
Consider $\mathbf{A}$ and $\mathbf{B}$ finite fermion systems described by the Clifford $\C$--algebras $\Q_{\mathbf{A}}$ and $\Q_{\mathbf{B}}$, respectively, with $\Nc\doteq|\Q_{\mathbf{A}}^{+}\cap\Q_{\mathbf{A}}|=|\Q_{\mathbf{B}}^{+}\cap\Q_{\mathbf{B}}|\in\N$. Take $\cal\doteq[0,1]$ and Gaussian states $\omega_{\mathbf{A}}\in\states_{\Q_{\mathbf{A}}}$ and $\omega_{\mathbf{B}}\in\states_{\Q_{\mathbf{B}}}$ satisfying assumptions of Lemma \ref{lemma:qfi}. Then under the beam--splitter quantum channel given by \eqref{eq:out_gaussian} below and $\lambda\in\cal$, the \emph{concave} entropy power inequality holds:
\begin{align*}
\E(\omega_{\mathbf{I}})\geq\lambda_{\mathbf{A}}\E(\omega_{\mathbf{A}})+\lambda_{\mathbf{B}}\E(\omega_{\mathbf{B}})\quad\text{with}\quad\lambda_{\mathbf{A}}\doteq\lambda,\lambda_{\mathbf{B}}\doteq1-\lambda. 
\end{align*}
\end{theorem} 
\begin{proof}
As is usual, we proceed in similar form that Blachman \cite{blachman}. For this, a few of supporting definitions are introduced:
\begin{enumerate}
\item[i.] Take $t\in\R_{0}^{+}$ and consider the differentiable functions $t_{\mathbf{A}},t_{\mathbf{B}}\in C^{1}(\R_{0}^{+};\R_{0}^{+})$ such that $t_{\mathbf{A}}(0)=t_{\mathbf{B}}(0)=0$, $t_{\mathbf{A}}(t)\approx\frac{t}{2}+\mathcal{O}(t)$ and $t_{\mathbf{B}}(t)\approx\frac{t}{2}+\mathcal{O}(t)$ as $t\to\infty$. This is physically justified because of the independence of the systems $\mathbf{A}$ and $\mathbf{B}$. See proof of Lemma \ref{lemma:stam}. 
\item[ii.] For $\mathbf{C}\in\{\mathbf{A},\mathbf{B},\mathbf{I}\}$, let $t_{\mathbf{C}}$ the time describing the evolution of the state $\rho_{\mathbf{C}}$. By Lemma \ref{lemma:stam}, Expression \eqref{eq:timesIAB} and i. one can define the function composition $t_{\mathbf{I}}(t)\doteq\lambda_{\mathbf{A}}t_{\mathbf{A}}(t)+\lambda_{\mathbf{B}}t_{\mathbf{B}}(t)$ for all $t\in\R_{0}^{+}$. 
\item[iii.] For $\mathbf{C}\in\{\mathbf{A},\mathbf{B},\mathbf{I}\}$, let $\E\left(\omega_{\mathbf{C},t_{\mathbf{C}}(t)}\right)$ be the quantum entropy power \eqref{eq:qep}, such that by the de Bruijin's identity, Lemma \ref{lemma:Bruijn}, we note that 
$$
\dot{\E}\left(\omega_{\mathbf{C},t_{\mathbf{C}}(t)}\right)\doteq\frac{\d}{\d t}\E\left(\omega_{\mathbf{C},t_{\mathbf{C}}(t)}\right)=\frac{1}{\Nc}\E\left(\omega_{\mathbf{C},t_{\mathbf{C}}(t)}\right)\mJ\left(\omega_{\mathbf{C},t_{\mathbf{C}}(t)}\right)\frac{\d}{\d t}t_{\mathbf{C}},
$$
Here we assume that $\E\left(\omega_{\mathbf{C},t_{\mathbf{C}}(0)}\right)=\E\left(\omega_{\mathbf{C}}\right)$. We can invoke the results \cite[Eqs. (19) and (84)]{palmargio14} and \cite[Corollary III-4]{konSmi14}, which are also valid for fermion systems. Namely, for any Gaussian state $\omega$ (with associated density matrix $\rho$) evolving in time by the Liouvillean \eqref{eq:ferm_liou} (that is $\rho_{t}=\e^{t\L}\rho$, for $t\in\R_{0}^{+}$), we have the following asymptotic estimate as $t\to\infty$
\begin{equation}
\mathcal{E}(\omega_{t})=\frac{\e}{2}t+\mathcal{O}(t)\label{eq:asymptotic}. 
\end{equation}
\item[iv.] For $\mathbf{C}\in\{\mathbf{A},\mathbf{B}\}$ the times satisfy the initial value problem $\dot{t}_{\mathbf{C}}\doteq\frac{\d}{\d t}t_{\mathbf{C}}(t)=\E\left(\omega_{\mathbf{C},t_{\mathbf{C}}(t)}\right)$, $\dot{t}_{\mathbf{C}}(0)=0$. Observe that by the \emph{Peano's Theorem} for ordinary differential equations we know that has at least one solution.
\item[v.] To lighten notations, for $\mathbf{C}\in\{\mathbf{A},\mathbf{B},\mathbf{I}\}$ let $\E_{\mathbf{C}}(t)\equiv\E\left(\omega_{\mathbf{C},t_{\mathbf{C}}(t)}\right)$ and $\mJ_{\mathbf{C}}(t)\equiv\mJ\left(\omega_{\mathbf{C},t_{\mathbf{C}}(t)}\right)$. Thus by following iii. $\E_{\mathbf{C}}(0)$ equals $\E(\omega_{\mathbf{C}})$.
\end{enumerate}
First of all, in the Stam inequality, Lemma \ref{lemma:stam}, take $\alpha=\frac{\mJ\left(\omega_{\mathbf{A}}\right)^{-1}}{\mJ\left(\omega_{\mathbf{A}}\right)^{-1}+\mJ\left(\omega_{\mathbf{B}}\right)^{-1}}$ and $\beta=\frac{\mJ\left(\omega_{\mathbf{B}}\right)^{-1}}{\mJ\left(\omega_{\mathbf{A}}\right)^{-1}+\mJ\left(\omega_{\mathbf{B}}\right)^{-1}}$: 
$$\frac{\lambda_{\textbf{A}}}{\mJ\left(\omega_{\mathbf{A}}\right)}+\frac{\lambda_{\textbf{B}}}{\mJ\left(\omega_{\mathbf{B}}\right)}\leq\frac{1}{\mJ\left(\omega_{\mathbf{I}}\right)}.
$$
Combining last inequality with the AM--GM inequality applied to $\E_{\mathbf{A}}(t)^{2}\mJ_{\mathbf{A}}(t)^{2}$ and $\mathcal{E}_{\mathbf{B}}(t)^{2}\mathcal{J}_{\mathbf{B}}(t)^{2}$, and some rearrangements we get
\begin{eqnarray}
\lambda_{\mathbf{A}}\mathcal{E}_{\mathbf{A}}(t)^{2}\mathcal{J}_{\mathbf{A}}(t)+\lambda_{\mathbf{B}}\mathcal{E}_{\mathbf{B}}(t)^{2}\mathcal{J}_{\mathbf{B}}(t)
&\geq&(\lambda_{\mathbf{A}}\mathcal{E}_{\mathbf{A}}(t)+\lambda_{\mathbf{B}}\mathcal{E}_{\mathbf{B}}(t))^{2}\frac{\mathcal{J}_{\mathbf{A}}(t)\mathcal{J}_{\mathbf{B}}(t)}{\lambda_{\mathbf{B}}\mathcal{J}_{\mathbf{A}}(t)+\lambda_{\textbf{A}}\mathcal{J}_{\mathbf{B}}(t)}\nonumber\\
&\geq&(\lambda_{\mathbf{A}}\mathcal{E}_{\mathbf{A}}(t)+\lambda_{\mathbf{B}}\mathcal{E}_{\mathbf{B}}(t))^{2}\mathcal{J}_{\mathbf{I}}(t).\label{eq:ineJE} 
\end{eqnarray}
With the previous notations we can study the behavior of the positive valued real differentiable function
$$
f_{\textbf{A},\textbf{B},\textbf{I}}(t)\doteq\frac{\lambda_{\textbf{A}}\E_{\mathbf{A}}(t)+\lambda_{\textbf{B}}\E_{\mathbf{B}}(t)}{\E_{\mathbf{I}}(t)}.
$$ 
From ii. iii. iv., v. and Inequality \eqref{eq:ineJE} note that $f_{\textbf{A},\textbf{B},\textbf{I}}'(t)\geq0$:
\begin{align*}
f_{\textbf{A},\textbf{B},\textbf{I}}'(t)=&\frac{\left(\lambda_{\textbf{A}}\dot{\E}_{\mathbf{A}}(t)+\lambda_{\textbf{B}}\dot{\E}_{\mathbf{B}}(t)\right)\E_{\mathbf{I}}(t)-\left(\lambda_{\textbf{A}}\E_{\mathbf{A}}(t)+\lambda_{\textbf{B}}\E_{\mathbf{B}}(t)\right)\dot{\E}_{\mathbf{I}}(t)}{\E_{\mathbf{I}}(t)^{2}}\\
=&\frac{\left(\lambda_{\textbf{A}}\E_{\mathbf{A}}(t)\mJ_{\mathbf{A}}(t)\dot{t}_{\mathbf{A}}+\lambda_{\textbf{B}}\E_{\mathbf{B}}(t)\mJ_{\mathbf{B}}(t)\dot{t}_{\mathbf{B}}\right)-\left(\lambda_{\textbf{A}}\E_{\mathbf{A}}(t)+\lambda_{\textbf{B}}\E_{\mathbf{B}}(t)\right)\mJ_{\mathbf{I}}(t)\dot{t}_{\mathbf{I}}}{\Nc\E_{\mathbf{I}}(t)}\\
=&\frac{\left(\lambda_{\textbf{A}}\E_{\mathbf{A}}(t)^{2}\mJ_{\mathbf{A}}(t)+\lambda_{\textbf{B}}\E_{\mathbf{B}}(t)^{2}\mJ_{\mathbf{B}}(t)\right)-\left(\lambda_{\textbf{A}}\E_{\mathbf{A}}(t)+\lambda_{\textbf{B}}\E_{\mathbf{B}}(t)\right)^{2}\mJ_{\mathbf{I}}(t)}{\Nc\E_{\mathbf{I}}(t)}.
\end{align*}
Hence, $f_{\textbf{A},\textbf{B},\textbf{I}}(t)$ is an increasing function. Thus $f_{\textbf{A},\textbf{B},\textbf{I}}(t\to\infty)\geq f_{\textbf{A},\textbf{B},\textbf{I}}(0)$, which according to i., ii.,  iii. and v. yields us to 
$$
1\geq\frac{\lambda_{\textbf{A}}\E_{\mathbf{A}}(0)+\lambda_{\textbf{B}}\E_{\mathbf{B}}(0)}{\E_{\mathbf{I}}(0)}=\frac{\lambda_{\textbf{A}}\E(\omega_{\mathbf{A}})+\lambda_{\textbf{B}}\E(\omega_{\mathbf{B}})}{\E(\omega_{\mathbf{I}})},
$$ 
concluding the proof of the Theorem. 
\end{proof}
\section{Technical results}\label{sec:tech_proofs}
\subsection{Beam splitter operators}
Consider the self--dual Hilbert space $(\H,\a)$, and let $\A'$ be the self--dual $\CAR$ algebra generated by the unit $\1$ and the family $\{\mathrm{C}(\varphi )\}_{\varphi \in \H}$ of elements satisfying: $\mathrm{C}\left(\varphi\right)^{*}$ is (complex) linear, $\mathrm{C}(\varphi )^{*}\doteq\mathrm{C}(\a(\varphi))$ for any $\varphi\in\H$ and satisfies the $\CAR$, \eqref{eq: CAR}. For the self--dual $\C$--algebra $\A$, with generators $\1$ and $\{\B(\varphi )\}_{\varphi \in \H}\in\A$, we consider the product $\C$--algebra $\V\equiv\A\otimes\A'$ such that for any $\varphi_{1},\varphi_{2}\in\H$, 
$$
\B(\varphi_{1})\otimes\1\equiv\B(\varphi_{1})\in\V,\qquad\1\otimes\mathrm{C}(\varphi_{2})\equiv\mathrm{C}(\varphi_{2})\in\V
$$ 
so that 
\begin{eqnarray}
\varphi_{1}\oplus\varphi_{2}\to\B(\varphi_{1})\mathrm{C}(\varphi_{2})^{\#}+\mathrm{C}(\varphi_{2})^{\#}\B(\varphi_{1})=0. \label{eq:BC}
\end{eqnarray}
Here the symbol $\#$ stands for either $\mathrm{C}(\varphi_{2})$ or $\mathrm{C}(\varphi_{2})^{*}$. Note that, due to all elements of $\A$ and $\A'$ are bounded, the product $\C$--algebra $\V$ is well--defined, see \cite{dixmier77}. We define:
\begin{definition}[Displacement and Beam Splitter operators on self--dual $\CAR$--algebras]\label{def:dis_weyl}
Let $(\H,\a)$ be a self--dual Hilbert space, and let $\A$ and $\A'$ be the self--dual $\CAR$ algebras as above defined. Fix a basis projection $P$ associated with $(\H,\a)$ and an orthonormal basis $\{\psi_{j}\}_{j\in J}$ of its range $\mathfrak{h}_{P}$. 
\begin{enumerate}
\item The element $\inner{\mathrm{C},\B}$ on $\V$ is defined by
$$
\inner{\mathrm{C},\B}\doteq\sum_{j\in J}\mathrm{C}(\psi_{j})^{*}\B(\psi_{j}),
$$
where the families of elements $\B_{|J|}\doteq\{\B(\psi_{i})\}_{j\in J}$, $\mathrm{C}_{|J|}\doteq\{\mathrm{C}(\psi_{i})\}_{j\in J}$ will describe $|J|$ different \emph{modes}. 
\item Denote by 
$$
\left(\mathbf{C}_{|J|}\right)^{\mathrm{T}}\doteq\left(\mathrm{C}(\psi_{1}),\mathrm{C}(\psi_{1})^{*},\ldots,\mathrm{C}(\psi_{|J|}),\mathrm{C}(\psi_{|J|})^{*}\right)\quad\text{and}\quad
\mJ\doteq
\left( \begin{array}{cc}
0 & 1 \\ 
1 & 0
\end{array}
\right)^{\oplus|J|}.
$$
The fermionic ``Weyl displacement operator'' $\mathbb{D}_{\A}\left(\mathbf{C}_{|J|}\right)\colon\A\to\V$, of $\A$ relative to $\A'$, is defined by
$$
\mathbb{D}_{\A}\left(\mathbf{C}_{|J|}\right)\doteq\e^{-\left(\mathbf{C}_{|J|},\mJ\mathbf{B}_{|J|}\right)_{\A}},
$$
with $(A,B)_{\A}\doteq A^{T}B$ for any $A,B\in\A,\A'$, and the exponential function given by \eqref{eq:expCAR}.
\item Denote by $\mathcal{C}$ the compact interval $[0,1]$. The beam splitter map $\mathbb{U}_{\lambda}\colon\mathcal{C}\to\V$ is the unitary operator on $\BL(\V)$ given by
$$
\mathbb{U}_{\lambda}\doteq\e^{f(\lambda)\inner{\B,\mathrm{C}}-f(\lambda)^{*}\inner{\mathrm{C},\B}},
$$
fot $\lambda\in\mathcal{C}$ the \emph{transmissivity} of the beam splitter, and $f\colon\mathcal{C}\to\CP\setminus\{0\}$ a well--defined complex function controlling the relative weight of the $\C$--algebras $\A$ and $\A'$, where $f(\lambda)^{*}\in\CP\setminus\{0\}$ is the conjugate complex of $f(\lambda)$. 
\end{enumerate}
\end{definition}
We can state the following \cite{konSmi14}:
\begin{lemma}\label{lemma:displ_weyl}
In regard to the displacement operator $\mathbb{D}_{\A}$ of Definition \ref{def:dis_weyl} we have the following properties:
\begin{enumerate}
\item For any $\inner{\B,\mathrm{C}}$ and $\inner{\mathrm{C},\B}$ as in Definition we have:
$$
\mathbb{D}_{\A}\left(\mathbf{C}_{|J|}\right)=\e^{\inner{\B,\mathrm{C}}-\inner{\mathrm{C},\B}},
$$
and 
$$
\mathbb{D}_{\A}\left(\mathbf{C}_{|J|}\right)^{-1}=\mathbb{D}_{\A}\left(-\mathbf{C}_{|J|}\right)=\mathbb{D}_{\A}\left(\mathbf{C}_{|J|}\right)^{*}=\mathbb{D}_{\A'}\left(\mathbf{B}_{|J|}\right)^{-1},
$$
\item For any $\lambda\in\mathcal{C}$, exists a complex matrix $\mathbb{M}_{\lambda}^{|J|}\in\mathrm{Mat}(2|J|,\CP)$ such that the \emph{Heisenberg evolution} of the \emph{modes} $\mathbf{B}$ and $\mathbf{C}$ is given by
$$
\mathbb{U}_{\lambda}^{*}
\mathbf{D}
  \mathbb{U}_{\lambda}=
  \mathbb{M}_{\lambda}^{|J|}\mathbf{D},
$$
where for $\mathbb{M}_{\lambda}\in\mathrm{Mat}(2,\CP)$, $\mathbb{M}_{\lambda}^{|J|}\doteq\mathbb{M}_{\lambda}^{\oplus|J|}$ and
\begin{align*}
\mathbf{D}_{|J|}^{\mathrm{T}}\doteq\left(\B(\psi_{1}),\mathrm{C}(\psi_{1}),\ldots,\B(\psi_{|J|}),\mathrm{C}(\psi_{|J|})\right). 
\end{align*}
\end{enumerate}
\end{lemma}
\begin{proof}
Part 1. of Lemma is straightforward from Definition \ref{def:dis_weyl}. For the part 2. it is enough consider the modes $\B(\psi_{1})$ and $\mathrm{C}(\psi_{1})$ such that we can use Expression \eqref{eq:oper_expCAR} to get
$$
\mathbb{U}_{\lambda}^{*}
\begin{pmatrix}
    \B(\psi_{1})\\
    \mathrm{C}(\psi_{1})
\end{pmatrix}
  \mathbb{U}_{\lambda}=
  \mathbb{M}_{\lambda}\begin{pmatrix}
    \B(\psi_{1})\\
    \mathrm{C}(\psi_{1})
\end{pmatrix},
$$
where $\mathbb{M}_{\lambda}\in\mathrm{Mat}(2,\CP)$ is a $2\times2$ complex matrix with components $\left(\mathbb{M}_{\lambda}\right)_{1,1}=\left(\mathbb{M}_{\lambda}\right)_{2,2}=\cos\left(\left|f\left(\lambda\right)\right|\right)$ and
$$
\left(\mathbb{M}_{\lambda}\right)_{1,2}=\sin\left(\left|f\left(\lambda\right)\right|\right)\sqrt{\frac{f\left(\lambda\right)}{f\left(\lambda\right)^{*}}}\quad\text{and}\quad \left(\mathbb{M}_{\lambda}\right)_{2,1}=-\sin\left(\left|f\left(\lambda\right)\right|\right)\sqrt{\frac{f\left(\lambda\right)^{*}}{f\left(\lambda\right)}}.
$$
Considering the matrix $\mathbb{M}_{\lambda}^{\oplus|J|}\in\mathrm{Mat}(2|J|,\CP)$, the proof follows.
\end{proof}
We interprete Part 2. of last Lemma saying that if $\A$ and $\A'$ describe two different fermionic systems interacting between them, then the beam splitter operator $\mathbb{U}_{\lambda}$ is used to obtain a family of output modes of size $|J|$. For technical purposes, and w.l.o.g., one usually take $\mathrm{ran}\left(\cos\left(\left|f\left(\lambda\right)\right|\right)\right)=\mathrm{ran}\left(\sin\left(\left|f\left(\lambda\right)\right|\right)\right)=\mathcal{C}$, so that for any $\lambda\in\mathcal{C}$, $\cos\left(\left|f\left(\lambda\right)\right|\right)=\sqrt{\lambda}$ and $\sin\left(\left|f\left(\lambda\right)\right|\right)=\sqrt{1-\lambda}$, with $f(\lambda)=\e^{\ii\theta}\sqrt{\left|f(\lambda)\right|}$ for some $\theta\in\R$. Then, we writte
$$
\mathbb{M}_{\lambda,\theta}\doteq \begin{pmatrix}
    \sqrt{\lambda} & \sqrt{1-\lambda}\e^{\ii\theta}\\
   -\sqrt{1-\lambda}\e^{-\ii\theta} & \sqrt{\lambda}
  \end{pmatrix},
$$
which is the well--known matrix implementing the beam--splitter operator. Here, for physical purposes we will asume that $\theta=0$, so that 
\begin{equation}
\mathbb{M}_{\lambda}\doteq\mathbb{M}_{\lambda,0}= 
\begin{pmatrix}
    \sqrt{\lambda} & \sqrt{1-\lambda}\\
   -\sqrt{1-\lambda} & \sqrt{\lambda}
\end{pmatrix}, \label{eq:matrixM}
\end{equation}
By taking into account the quantum channel given by \eqref{eq:map_red_sys} for the unitary operator $\mathbb{U}_{\lambda}$ we interprete by
\begin{equation}
\rho_{\mathbf{I}}^{(\lambda)}\doteq\E_{\mathbb{U}_{\lambda}}(\rho_{\mathbf{A}}\otimes\rho_{\mathbf{B}})=\tr_{\mathbf{B}}\left(\mathbb{U}_{\lambda}^{*}\left(\rho_{\mathbf{A}}\otimes\rho_{\mathbf{B}}\right)\mathbb{U}_{\lambda}\right)\label{eq:out_gaussian}
\end{equation}
the density matrix associated to the output state $\omega_{\mathbf{I}}\in\states_{\V}$.\par\bigskip
In regard to the Grassmann algebra $\wedge^{*}\H$ one can introduce Weyl displacement operators similarly to the self--dual $\CAR$ algebra $\A$ case (Definition \ref{def:dis_weyl}):
\begin{definition}[Weyl displacement operator at Grassmann algebras]\label{def:bil_weyl}
Let $(\H,\a)$ be a self--dual Hilbert space:
\begin{enumerate}
\item Fix a basis projection $P\in\fp$ and an orthonormal basis $\{\psi _{j}\}_{j\in J}$ of its range $\mathfrak{h}_{P}$. Given $k,l\in \N_{0}$, 
\begin{align*}
\inner{\mathfrak{h}_{P}^{(k)},\mathfrak{h}_{P}^{(l)}}\doteq\sum\limits_{j\in J}\left( \a\psi _{j}\right) ^{(k)}\wedge \psi_{j}^{(l)}.
\end{align*}
\item For $l\in\N_{0}$, define
$$
\left(\mathbf{\psi}_{|J|}^{(l)}\right)^{\mathrm{T}}\doteq\left(\psi_{1}^{(l)},\left(\a\psi_{1}\right)^{(l)},\ldots,\psi_{|J|}^{(l)},\left(\a\psi_{|J|}\right)^{(l)}\right)\quad\text{and}\quad
\mJ\doteq
\left( \begin{array}{cc}
0 & 1 \\ 
1 & 0
\end{array}
\right)^{\oplus|J|}.
$$
For any $k\in\N_{0}$ and $\mathbf{\psi}_{|J|}^{(k)}$, the fermionic ``Weyl displacement operator'' $\mathbb{T}_{k}\colon\wedge^{*}\h_{P}^{(k)}\to\wedge^{*}(\h_{P}^{(k)}\oplus\h_{P}^{*(l)})$, of the Hilbert space $\h_{P}^{(k)}$ relative to the Hilbert space $\h_{P}^{(l)}$, is defined by
$$
\mathbb{T}_{k}\left(\mathbf{\psi}_{|J|}^{(l)}\right)\doteq\e^{-\left(\mathbf{\psi}_{|J|}^{(l)},\mJ\mathbf{\psi}_{|J|}^{(k)}\right)_{\h_{P}}},
$$
with $(A,B)_{\h_{P}}\doteq A^{T}B$ for any $A,B\in\h_{P}^{(k)},\h_{P}^{(l)}$, and the exponential function of Expression \eqref{eq:expCAR}.
\end{enumerate}
\end{definition}
Straightforward calculations using Definition \ref{def:bil_weyl}--(2), for a basis projection $P\in\fp$, show that the Weyl displacement operator of Definition \ref{def:bil_weyl}--(3) can be redefined by
\begin{equation}\label{eq:displ1}
\mathbb{T}_{k}\left(\mathbf{\psi}_{|J|}^{(l)}\right)=\e^{\inner{\h_{P}^{(k)},\h_{P}^{(l)}}-\inner{\h_{P}^{(l)},\h_{P}^{(k)}}},
\end{equation}
so that
\begin{equation}\label{eq:displ2}
\mathbb{T}_{k}\left(\mathbf{\psi}_{|J|}^{(l)}\right)^{-1}=\mathbb{T}_{k}\left(-\mathbf{\psi}_{|J|}^{(l)}\right)=\mathbb{T}_{k}\left(\mathbf{\psi}_{|J|}^{(l)}\right)^{*}=\mathbb{T}_{l}\left(\mathbf{\psi}_{|J|}^{(k)}\right)^{-1},
\end{equation}
see Lemma \ref{lemma:displ_weyl} for a comparation with the Weyl displacement operator $\mathbb{D}_{\A}$ in the context of self--dual $\CAR$ algebras.\\
For a fix basis projection, consider the Grassmann $\C$--algebra $\G_{P}$. For the displacement operator $\mathbb{T}$ of Definition \ref{def:bil_weyl} we have:
\begin{lemma}\label{lemma:Weyl_grass}
Fix a basis projection $P\in\fp$ with range $\h_{P}$, and take same notations of Definition \ref{def:bil_weyl}. For the displacement operator $\mathbb{T}$ we have:
$$
\mathbb{T}_{k}\left(\mathbf{\psi}_{|J|}^{(l)}\right)^{*}\circ_{P}^{(k)}\left(\psi_{i}^{(k)}\right)^{\#}\circ_{P}^{(k)}\mathbb{T}_{k}\left(\mathbf{\psi}_{|J|}^{(l)}\right)=\left(\psi_{i}^{(k)}\right)^{\#}+\left(\psi_{i}^{(l)}\right)^{\#},
$$
for each $i\in J$, where $\circ_{P}^{(k)}$ is the circle product of Definition \ref{definition star} acting on $\wedge^{*}\h^{(k)}$. The symbol $\#$ stands for either $\psi_{j}^{(k,l)}$ or $\left(\a\psi_{j}\right)^{(k,l)}$. Thus for a fix Hilbert space $\h_{P}^{(k)}$, the element $\mathbf{\psi}_{|J|}^{(l)}$ is \emph{displaced} as
\begin{align*}
\mathbb{T}_{k}\left(\mathbf{\psi}_{|J|}^{(l)}\right)^{*}\circ_{P}^{(k)}\mathbf{\psi}_{|J|}^{(k)}\circ_{P}^{(k)}\mathbb{T}_{k}\left(\mathbf{\psi}_{|J|}^{(l)}\right)=\mathbf{\psi}_{|J|}^{(k)}+\mathbf{\psi}_{|J|}^{(l)}.
\end{align*}
Let $\mathbf{\psi}_{|J|}^{(l)}$ and $\mathbf{\psi}_{|J|}^{(m)}$ well--defined on $\wedge^{*}\h_{P}^{(l)}$ and $\wedge^{*}\h_{P}^{(m)}$, respectively. We have
\begin{align*}
\mathbb{T}_{k}\left(\mathbf{\psi}_{|J|}^{(l)}+\mathbf{\psi}_{|J|}^{(m)}\right)=
\mathbb{T}_{k}\left(\mathbf{\psi}_{|J|}^{(l)}\right)\circ_{P}^{(k)}\mathbb{T}_{k}\left(\mathbf{\psi}_{|J|}^{(m)}\right)\e^{-\frac{1}{2}\left(\mathbf{\psi}_{|J|}^{(l)},\mJ\mathbf{\psi}_{|J|}^{(m)}\right)}.
\end{align*}
\end{lemma}
\begin{remark}
Observe that this Lemma brings us to similar results proven for fermionic coherent states \cite{Ohnuki_Kashiwa,cahill99,com12coher}. The fundamental difference between our result and those of the mentioned works, is that here we do not require a $\CAR$ algebra and the \emph{Grassmann numbers}, which are usually introduced to provide anticommutative properties. Instead, the circle product $\circ_{P}$ on Grassmann algebras provides a \emph{natural} structure that combines the $\CAR$ algebra structure, Equation \eqref{eq:CAR_gras}, and the anticommutative property of Grassmann algebras, see \eqref{grassmana anticommute00}. Thus, it is only necessary having Grassmann algebras endowed with $\circ_{P}$ in order to study coherent states of fermions.
\end{remark}
\begin{proof}
Take $k,l\in\N_{0}$ and consider Expressions \eqref{eq:displ1}--\eqref{eq:displ2}. Similar to Proof of Lemma \ref{lemma:displ_weyl}, we can use Expressions \eqref{eq:oper_expCAR}--\eqref{eq:iden_fourtermsCAR} for the vector space $\wedge^{*}(\h_{P}^{(k)}\oplus\h_{P}^{(l)})$. Here, we are displacing a fix vector $\mathbb{\psi}_{|J|}^{(k)}$ the \emph{quantity} $\mathbb{\psi}_{|J|}^{(l)}$. Hence, for $\mathbb{\psi}^{(k)}_{|J|}$ and $\psi_{i}^{(l)}$ of Lemma, we take into account the circle product $\circ_{P}$ of Definition \ref{definition star}, Expressions \eqref{eq:displ1}--\eqref{eq:displ2} and \eqref{eq:oper_expCAR} in order to obtain
\begin{eqnarray*}
\mathbb{T}_{k}\left(\mathbf{\psi}_{|J|}^{(l)}\right)^{*}\circ_{P}^{(k)}\psi_{i}^{(k)}\circ_{P}^{(k)}\mathbb{T}_{k}\left(\mathbf{\psi}_{|J|}^{(l)}\right)&=&
\e^{\inner{\h_{P}^{(l)},\h_{P}^{(k)}}-\inner{\h_{P}^{(k)},\h_{P}^{(l)}}}\circ_{P}^{(k)}\psi_{i}^{(k)}\circ_{P}^{(k)}\e^{\inner{\h_{P}^{(k)},\h_{P}^{(l)}}-\inner{\h_{P}^{(l)},\h_{P}^{(k)}}}\\
&=&\psi_{i}^{(k)}+\left[A^{(k,l)},\psi_{i}^{(k)}\right]_{P}^{(k)}+\frac{1}{2}\left[A^{(k,l)},\left[A^{(k,l)},\psi_{i}^{(k)}\right]_{P}^{(k)}\right]_{P}^{(k)}\ldots, 
\end{eqnarray*}
with 
$$
A^{(k,l)}\doteq\sum_{j\in J}\left(\left(\a\psi_{j}\right)^{(l)}\psi_{j}^{(k)}-\left(\a\psi_{j}\right)^{(k)}\psi_{j}^{(l)}\right),
$$
and $[A,B]_{P}^{(k)}$ is the commutator of $A,B\in\wedge^{*}(\h_{P}^{(k)}\oplus\h_{P}^{*(l)})$ obeying the circle product $\CAR$ \eqref{eq:CAR_gras} for $A,B\in\wedge^{*}\h_{P}^{(k)}$, while it is equals to zero for $A,B\in\wedge^{*}\h_{P}^{(l)}$ according to \eqref{grassmana anticommute00}. Observe that for any basis projection $P\in\fp$ and integers $n\in\mathbb{N}_{0}$ and $k\in \{0,\ldots ,n\}$ we can define
\begin{equation}
\varkappa _{P}^{(k)}\doteq \varkappa _{(0,0)}^{(k,k)}\circ \varkappa _{P}
\label{definition xi-k}
\end{equation}
For $n\in\N_{0}$ we can extend the isomorphism $\varkappa_{P}$ of Definition \ref{def:isomor} for any $k$--copy $\wedge^{*}\H^{(k)}$ of $\wedge^{*}\H$. Thus for any $\varphi_{1},\varphi_{2}\in\H$, we have $\varphi_{1}^{(k)},\varphi_{2}^{(k)}\in\H^{(k)}$ satisfying \eqref{eq:CAR_gras}. Then, we can use this and Expression \eqref{eq:iden_fourtermsCAR} in order to calculate $\left[A^{(k,l)},\psi_{i}^{(l)}\right]_{P}^{(k)}$ as
$$
\left[A^{(k,l)},\psi_{i}^{(k)}\right]_{P}^{(k)}=\sum_{j\in J}\left\{\left(\a\psi_{j}\right)^{(k)}\circ_{P}^{(k)}\psi_{i}^{(k)}+\psi_{i}^{(k)}\circ_{P}^{(k)}\left(\a\psi_{j}\right)^{(k)}\right\}\psi_{j}^{(l)}.
$$
Hence we are able to combine Expressions \eqref{eq:CAR_gras}--\eqref{eq:prod_P}, from which one has $\left[A^{(k,l)},\psi_{i}^{(k)}\right]_{P}=\psi_{i}^{(l)}$, and hence by \eqref{grassmana anticommute00} and notation \eqref{eq:edge_product} we obtain
$$
\left[A^{(k,l)},\left[A^{(k,l)},\psi_{i}^{(k)}\right]_{P}^{(k)}\right]_{P}^{(k)}=-\left[\left(\a\psi_{j}\right)^{(l)}\wedge\psi_{i}^{(l)}+\psi_{i}^{(l)}\wedge\left(\a\psi_{j}\right)^{(l)}\right]\psi_{i}^{(k)}=0.
$$ 
This shows that the Lemma works for $\psi_{j}^{(k)}$, whereas a similar calculation shows that also works for $\left(\a\psi_{j}\right)^{(k)}$.\\
For the second part, note that for any $k,l\in\N_{0}$, all the elements of the form $\inner{\h_{P}^{(k)},\h_{P}^{(l)}}\in\wedge^{*}(\h_{P}^{(k)}\oplus\h_{P}^{*(l)})$ are even, which are in a commutative subalgebra (see Expression \eqref{grassmana anticommute00} and comments around it, as well as \eqref{def norm}). Thus for $k,l,m,n,o,p\in\N_{0}$ straightforward calculations arrive us to
$$
\left[\inner{\h_{P}^{(k)},\h_{P}^{(l)}},\left[\inner{\h_{P}^{(m)},\h_{P}^{(n)}},\inner{\h_{P}^{(o)},\h_{P}^{(p)}}\right]\right]=0.
$$
Then we apply \eqref{eq:iden_fourtermsCAR} and \eqref{eq:CBHT}, as well as a simple computations in order to obtain
\begin{eqnarray*}
\mathbb{T}_{k}\left(\mathbf{\psi}_{|J|}^{(l)}+\mathbf{\psi}_{|J|}^{(m)}\right)&=&\mathbb{T}_{k}\left(\mathbf{\psi}_{|J|}^{(l)}\right)\circ_{P}^{(k)}\mathbb{T}_{k}\left(\mathbf{\psi}_{|J|}^{(m)}\right)\e^{\frac{1}{2}\left[\inner{\h_{P}^{(l)},\h_{P}^{(k)}}-\inner{\h_{P}^{(k)},\h_{P}^{(l)}},\inner{\h_{P}^{(m)},\h_{P}^{(k)}}-\inner{\h_{P}^{(k)},\h_{P}^{(l)}}\right]}\\
&=&\mathbb{T}_{k}\left(\mathbf{\psi}_{|J|}^{(l)}\right)\circ_{P}^{(k)}\mathbb{T}_{k}\left(\mathbf{\psi}_{|J|}^{(m)}\right)\e^{-\frac{1}{2}\left(\inner{\h_{P}^{(l)},\h_{P}^{(m)}}-\inner{\h_{P}^{(m)},\h_{P}^{(l)}}\right)},
\end{eqnarray*}
which is equivalent to the desired identity.
\end{proof}\par
For $\J$ defined by \eqref{eq:finite_setJ}, consider the Clifford $\C$--algebra $\Q'\equiv(\Q',+,\cdot,^{*},\Vert\cdot\Vert_{\Q})$ with generators, the unit $\1$ and the family of self--adjoint elements $\{S_{\j}\}_{\j\in\J}$, and satisfying \eqref{eq:CAR_clif} so that $\dim\Q'=\dim\Q=2^{\dim\H}$. The product $\C$--algebra $\V'\equiv\Q\otimes\Q'$ is such a one obeying for any $\j\in\J$
$$
R_{\j}\otimes\1\equiv R_{\j}\in\V',\qquad\1\otimes S_{\j}\equiv S_{\j}\in\V'
$$ 
so that for any $\mathfrak{i},\j\in\J$
\begin{eqnarray}
R_{\j}S_{\mathfrak{i}}+S_{\mathfrak{i}}R_{\j}=0, \label{eq:RS}
\end{eqnarray}
cf. \eqref{eq:BC}. Because the isomorphism between Clifford $\C$--algebras and self--dual $\CAR$ algebras, the product $\C$--algebra $\V'$ is well--defined as well as an exponential function on $\V'$, cf. \eqref{eq:expCAR}. For $\theta\in\R$, one can \emph{displace} the operator $A\in\V'$ via the $^{*}$--automorphism $\Delta_{\theta}$ on $\V'$ defined by
\begin{equation}\label{eq:displa_clif}
\Delta_{\theta}(A)\doteq\e^{\frac{\theta}{2}\inner{S,R}}A\e^{-\frac{\theta}{2}\inner{S,R}},
\end{equation}
with $\inner{S,R}\doteq\sum\limits_{\j\in\J}S_{\j}R_{\j}\in\V'$, cf. displacement operator of Definition \ref{def:dis_weyl} and Lemma \ref{lemma:displ_weyl}. By \eqref{eq:oper_expCAR}, we are able to write
$$
\e^{\frac{\theta}{2}\inner{S,R}}A\e^{-\frac{\theta}{2}\inner{S,R}}=A+\sum_{n=1}^{\infty}\frac{\theta^{n}}{2^{n}n!}\text{ad}_{\inner{S,R}}^{n}(A)
$$
where for $n\in\N$, $\text{ad}_{\inner{S,R}}^{n}(A)\doteq[\inner{S,R},[\inner{S,R},[\ldots,A]]\ldots]]$ is the $n$--fold commutator of $A$ with $\inner{S,R}$. In particular, for $A\equiv R_{\mathfrak{i}}\in\Q$, with $\mathfrak{i}\in\J$
$$
\e^{\frac{\theta}{2}\inner{S,R}}R_{\mathfrak{i}}\e^{-\frac{\theta}{2}\inner{S,R}}=\cos\left(\theta\right)R_{\mathfrak{i}}+\sin\left(\theta\right) S_{\mathfrak{i}}.
$$
By taking $\cos\left(\theta_{\lambda}\right)=\sqrt{\lambda}$ and $\sin\left(\theta_{\lambda}\right)=\sqrt{1-\lambda}$, for $\lambda\in(0,1)$ we can write
\begin{equation}
\e^{\frac{\theta_{\lambda}}{2}\inner{S,R}}R_{\mathfrak{i}}\e^{-\frac{\theta_{\lambda}}{2}\inner{S,R}}=\sqrt{\lambda}R_{\mathfrak{i}}+\sqrt{1-\lambda}S_{\mathfrak{i}}\qquad \theta_{\lambda}\doteq\arctan\left(\sqrt{\frac{1-\lambda}{\lambda}}\right).\label{eq:dis_weyl_cliff} 
\end{equation}
I.e., by using the $^{*}$--automorphism \eqref{eq:displa_clif} one arrives to a fermionic quantum version of the addition rule at the \emph{classical phase space}. See \cite[Eqs. (4) and (6)]{konSmi14}.
\subsection{Dissipative systems and Clifford algebras}
Now, suppose that the density matrix $\rho\in\Q^{+}\cap\Q$ satisfies the quantum diffusion equation
\begin{equation}
\frac{\d}{\d t}\rho_{t}=\L\rho_{t},\qquad \rho_{t}\doteq\e^{t\L}\rho,\label{eq:qde} 
\end{equation}
for any $t\in\R_{0}^{+}$, where by definition $\rho_{0}\doteq\rho$. See \eqref{eq:Liouvillean} below. Here, $\L\in\Q$ is the infinitesimal generator or Liouvillean of the strongly continuous semigroup $\{\e^{t\L}\}_{t\in\R_{0}^{+}}$, that we will assume to be bounded, and satisfies $\L\1=0$. Explicitly, using \eqref{eq:Lindblad}, with $V_{\j}=R_{\j}$, for any $A\in\Q$, $\L$ is given by 
\begin{equation}
\L A=\sum_{\j\in\J}\left(R_{\j}[A,R_{\j}]+[R_{\j},A]R_{\j}\right)=2\sum_{\j\in\J}\left(R_{\j}AR_{\j}-A\right)=-\sum_{\j\in\J}\left[R_{\j},\left[R_{\j},A\right]\right],\label{eq:ferm_liou} 
\end{equation}
where we use \eqref{eq:CAR_clif}. Then, we can recognize by $\L$, the generator of the infinite--temperature \emph{Fermi Ornstein--Uhlenbeck semigroup}. For details see \cite{carlen2020non}. By \eqref{eq:number_gross2}, for any even element $A$ (see \eqref{eq:even odd}) one notes that the Liouvillean and the fermionic number operators are related by $\Nc=-\frac{1}{4}\L$. Then, for any $A_{1},A_{2}\in\Q^{+}\cap\Q$ even elements of $\Q$ we get:
\begin{equation}
\inner{\L A_{1},A_{2}}_{\Q}^{\text{H.S.}}=\inner{A_{1},\L A_{2}}_{\Q}^{\text{H.S.}},\label{eq:HSevenQ} 
\end{equation}
where $\inner{\cdot,\cdot}_{\Q}^{\text{H.S.}}$ denotes the Hilbert--Schmidt inner product on $\Q$ given by \eqref{eq:HSclifford}. One can verify that the strongly continuous semigroup $\{\e^{t\L}\}_{t\in\R_{0}^{+}}$ with the Liouvillean $\L$ defined by \eqref{eq:ferm_liou} is trace--preserving completely positive \cite{gorini76}, it follows that $\L$ defines a quantum channel map. For the special case of the Clifford algebra $\Q$, given some density matrix $\rho\in\Q^{+}\cap\Q$ defining a Gaussian state $\omega_{C}$, the skew--symmetric matrix covariance $C\equiv C_{\rho}$ is given by \cite{bravyi05}
\begin{equation}\label{eq:matrixCrho}
C_{\mathfrak{i},\j}\doteq
\inner{\a\psi_{\mathfrak{i}},C\psi_{\j}}_{\H}=\frac{\ii}{2}\tr_{\Q}\left(\rho[R_{\mathfrak{i}},R_{\j}]\right),
\end{equation}
for $\mathfrak{i},\j\in\J$, and satisfying $C^{\mathrm{t}}C\leq\mathrm{1}_{|\J|}$, where $\mathrm{1}_{|\J|}\in\mathrm{Mat}(|\J|,\CP)$ denotes the identity matrix, and the symbol $\mathrm{t}$ denotes the transpose matrix, see \eqref{eq:finite_setJ}--\eqref{eq:CAR_clif}. See also \eqref{eq:skewC} and comments around it. In above inequality, $C^{\mathrm{t}}C=\mathrm{1}_{|\J|}$ holds for \emph{pure} Gaussian states, that is equivalent to say that $\lambda_{j}=\pm1$ for $j\in\frac{1}{2}|\J|$ in \eqref{eq:skewLambda}, \cite{bravyi05}. As proven by Bravyi--K\"{o}nig \cite[Lemma 1]{bravyi12clas}, for an \emph{initial} Gaussian state $\omega$ with density matrix $\rho$ satisfying the differential equation \eqref{eq:qde}, the Liouvillean \eqref{eq:ferm_liou} preserves the Gaussianity of the state $\omega_{t}$ (associated to $\rho_{t}\doteq\e^{t\L}\rho$) for all non--negative times $t\in\R_{0}^{+}$. 
\begin{lemma}[Entropy variation rate]\label{lemma:qfi}
Consider a Gaussian state $\omega\in\states_{\Q}$ with associated density matrix $\rho\equiv\rho_{\omega}\in\Q^{+}\cap\Q$. Suppose that the family of generators $\{R_{\j}\}_{\j\in\J}$ of $\Q$ satisfies $\sup\limits_{\j\in\J}\left\Vert R_{\j}\rho\right\Vert_{\Q}\in\R^{+}$\footnote{Then, $\Q$ is not necessarily finite dimensional, we only require its separability.\label{footnote:separability}}. Then, for any $\j\in\J$, the quantum Fisher information is given by
\begin{equation}
    \mJ\left(\omega_{R_{\j}}\right)=-\tr_{\Q}(\rho[R_{\j},[R_{\j},\ln\rho]]),
\end{equation}
in such a way that the entropy variation rate, see \eqref{eq:entropy_var_rate}, is  
\begin{align*}
\mJ\left(\omega\right)=-\sum_{\j\in\J}\tr_{\Q}(\rho[R_{\j},[R_{\j},\ln\rho]]). 
\end{align*}
\end{lemma}
\begin{proof}
Fix $R_{\j}$, with $\j\in\J$, and the assumptions of Lemma. Note that for any $\theta,\varepsilon\in\R^{+}$, we are able to write
$$
\S(\omega\Vert\omega_{R_{\j}}^{(\theta+\varepsilon)})-\S(\omega\Vert\omega_{R_{\j}}^{(\theta)})=
\tr_{\Q}\left(\rho\left(\ln\left(\rho_{R_{\j}}^{(\theta)}\right)-\ln\left(\rho_{R_{\j}}^{(\theta+\varepsilon)}\right)\right)\right),
$$
where $\S(\cdot,\cdot)$ is the relative entropy between two states given by \eqref{eq:rel_entro} and $\rho_{R_{\j}}^{(\theta)}$ is the displacement of $\rho$ according to \eqref{eq:displa_clif2} for $\theta\in\R$. Since $\rho\in\Q^{+}\cap\Q$ we can use the identity \eqref{eq:func_ident1}, in order to obtain
\begin{align}
\ln\left(\rho_{R_{\j}}^{(\theta)}\right)-\ln\left(\rho_{R_{\j}}^{(\theta+\varepsilon)}\right)=
\int_{0}^{\infty}\left(x\1+\rho_{R_{\j}}^{(\theta+\varepsilon)}\right)^{-1}\left(\rho_{R_{\j}}^{(\theta)}-\rho_{R_{\j}}^{(\theta+\varepsilon)}\right)\left(x\1+\rho_{R_{\j}}^{(\theta)}\right)^{-1}\d x.\label{eq:pivotal_qfi1}
\end{align}
Define the quantity 
$$
T_{\theta,\epsilon,\rho_{R_{\j}}}\doteq-\tr_{\Q}\left(\int_{0}^{\infty}\left(x\1+\rho_{R_{\j}}^{(\theta+\varepsilon)}\right)^{-1}\left(\rho_{R_{\j}}^{(\theta)}-\rho_{R_{\j}}^{(\theta+\varepsilon)}\right)\left(x\1+\rho_{R_{\j}}^{(\theta)}\right)^{-1}\d x\right),
$$
so that we desire to verify that the limit $\lim\limits_{\varepsilon\to0^{+}}\frac{|T_{\theta,\epsilon,\rho_{R_{\j}}}|}{\varepsilon}$ is bounded. By using similar arguments that in \cite[Example 6.2.31]{BratteliRobinson} we can show that
$$
\left|T_{\theta,\epsilon,\rho_{R_{\j}}}\right|\leq \left\Vert\left(\rho_{R_{\j}}^{(\theta+\varepsilon)}\right)^{-1}\right\Vert_{\Q}\left\Vert\rho_{R_{\j}}^{(\theta)}-\rho_{R_{\j}}^{(\theta+\varepsilon)}\right\Vert_{\Q},
$$
and using the identity
\begin{equation}
\rho_{R_{\j}}^{(\theta)}-\rho_{R_{\j}}^{(\theta+\varepsilon)}=\e^{\theta R_{\j}}\rho\e^{-\theta R_{\j}}\left(1-\e^{-\varepsilon R_{\j}}\right)-\e^{\theta R_{\j}}\left(\e^{\varepsilon R_{\j}}-1\right)\rho\e^{-(\theta+\varepsilon)R_{\j}},\label{eq:pivotal_qfi2}
\end{equation}
we get
$$
\lim_{\varepsilon\to0^{+}}\frac{\left|T_{\theta,\epsilon,\rho_{R_{\j}}}\right|}{\varepsilon}\leq2 \left\Vert\rho^{-1}\right\Vert_{\Q}\left\Vert R_{\j}\rho\right\Vert_{\Q},
$$
which, by the hypothesis of the Lemma ensures the boundedness of the limit $\lim\limits_{\varepsilon\to0^{+}}\frac{|T_{\theta,\epsilon,\rho_{R_{\j}}}|}{\varepsilon}$. On the other hand, note that
\begin{align}
\lim\limits_{\varepsilon\to0}\frac{\ln\left(\rho_{R_{\j}}^{(\theta)}\right)-\ln\left(\rho_{R_{\j}}^{(\theta+\varepsilon)}\right)}{\varepsilon}=
\int_{0}^{\infty}\left(x\1+\rho_{R_{\j}}^{(\theta)}\right)^{-1}\left(\rho_{R_{\j}}^{(\theta)}R_{\j}-R_{\j}\rho_{R_{\j}}^{(\theta)}\right)\left(x\1+\rho_{R_{\j}}^{(\theta)}\right)^{-1}\d x,\label{eq:pivotal_qfi3}
\end{align}
in particular, last expression can be rewritten as
\begin{align}
\int_{0}^{\infty}\left[R_{\j}\left(\left(x\1+\rho_{R_{\j}}^{(\theta)}\right)^{-1}-\left(x\1+\1\right)^{-1}\right)-\left(\left(x\1+\rho_{R_{\j}}^{(\theta)}\right)^{-1}-\left(x\1+\1\right)^{-1}\right)R_{\j}\right]\d x\nonumber\\
=\ln(\rho_{R_{\j}}^{(\theta)})R_{\j}-R_{\j}\ln(\rho_{R_{\j}}^{(\theta)}),\label{eq:pivotal_qfi4}
\end{align}
where we had used the identity \eqref{eq:func_ident2}. It follows combining \eqref{eq:pivotal_qfi1}, \eqref{eq:pivotal_qfi2}, \eqref{eq:pivotal_qfi3},  \eqref{eq:pivotal_qfi4} and the \emph{Lebesgue's dominated convergence theorem} that
$$
\frac{\d}{\d\varepsilon}
\S(\omega\Vert\omega_{R_{\j}}^{(\theta+\varepsilon)})=
\tr_{\Q}\left(\rho\left[\ln\left(\rho_{R_{\j}}^{(\theta)}\right),R_{\j}\right]\right).
$$
Similarly, by taking into account \eqref{eq:pivotal_qfi1}, \eqref{eq:pivotal_qfi2}, \eqref{eq:pivotal_qfi3}, \eqref{eq:pivotal_qfi4} and the Lebesgue's dominated convergence theorem one more time we obtain
\begin{align*}
\lim\limits_{\theta\to0}
\int_{0}^{\infty}\left[R_{\j}\left(\left(x\1+\rho_{R_{\j}}^{(\theta)}\right)^{-1}-\left(x\1+\rho\right)^{-1}\right)-\left(\left(x\1+\rho_{R_{\j}}^{(\theta)}\right)^{-1}-\left(x\1+\rho\right)^{-1}\right)R_{\j}\right]\d x&\\
= R_{\j}\ln\left(\rho\right)R_{\j} -R_{\j}^{2}\ln\left(\rho\right) -\ln\left(\rho\right)R_{\j}^{2}+R_{\j}\ln\left(\rho\right)R_{\j}
=-[R_{\j},[R_{\j},\ln\left(\rho\right)]]&,
\end{align*}
from which we deduce for any $R_{\j}$ that the quantum Fisher information \eqref{eq:qfi} is explicitly given by 
\begin{align*}
\mJ(\omega_{R_{\j}})=-\tr_{\Q}\left(\rho[R_{\j},[R_{\j},\ln(\rho)]]\right),
\end{align*}
and the conclusion follows.
\end{proof}
Then, we are in a position to state a fermionic version of the \emph{de Bruijin identity}:
\begin{lemma}[De Bruijin's Identity]\label{lemma:Bruijn}
Let $\L$ be the \emph{fermionic} Liouvillean \eqref{eq:ferm_liou}, and let $\omega\in\states_{\Q}$ be a Gaussian state with associated density matrix $\rho\equiv\rho_{\omega}\in\Q^{+}\cap\Q$ satisfying the differential equation \eqref{eq:qde}, and such that $\left\Vert\rho_{t}^{-1}\right\Vert_{\Q},\left\Vert\frac{\d\rho_{t}}{\d t}\right\Vert_{\Q}\in\R^{+}$ (see footnote \ref{footnote:separability}). Then, for any $t\in\R_{0}^{+}$, the fermionic \emph{de Bruijin} identity holds
\begin{align*}
\frac{\d}{\d t}\mS(\omega_{t})=\mJ(\omega_{t}),
\end{align*}
where $\mS(\omega_{t})$ and $\mJ(\omega_{t})$ are the the von Neumann entropy \eqref{eq:vNentro} and the entropy variation rate \eqref{eq:entropy_var_rate} of the state $\omega_{t}$ (provided $\rho_{t}$), respectively.
\end{lemma}
\begin{proof}
For the sake of completeness for a density matrix $\rho\in\Q^{+}\cap\Q$ we first prove the following well--known relation between the Liouvillean $\L$ and the entropy $\mS(\rho_{t})\equiv\mS(\omega_{t})$ \cite{spohn78}: 
\begin{equation}
 \frac{\d}{\d t}\mS(\rho_{t})=-\tr_{\Q}(\L(\rho_{t})\ln\rho_{t}),\label{eq:entropy_liou}
\end{equation}
with $\rho_{t}=\e^{t\L}\rho$ satisfying the quantum diffusion equation \eqref{eq:qde}, for all $t\in\R_{0}^{+}$, $\mS(\rho_{t})\doteq-\tr_{\Q}(\rho_{t}\ln\rho_{t})$ and the Liouvillean \eqref{eq:ferm_liou}. In fact, consider the difference between the entropies $\mS(\rho_{t+\varepsilon})$ and $\mS(\rho_{t})$, with $\varepsilon\in\R^{+}$, that is, 
\begin{multline}
\mS(\rho_{t+\varepsilon})-\mS(\rho_{t})=-\tr_{\Q}(\rho_{t+\varepsilon}\ln\rho_{t+\varepsilon})+\tr_{\Q}(\rho_{t}\ln\rho_{t})\\
 =-\tr_{\Q}\left(\left(\rho_{t+\varepsilon}-\rho_{t}\right)\ln(\rho_{t+\varepsilon})
+\rho_{t}\int_{0}^{\infty}\left(x\1+\rho_{t+\varepsilon}\right)^{-1}\left(\rho_{t}-\rho_{t+\varepsilon}\right)\left(x\1+\rho_{t}\right)^{-1}\d x\right),
\end{multline}
where we have used some rearrangements and we took the identity \eqref{eq:func_ident1}. Define the quantity 
$$
T_{\rho_{t},\epsilon}\doteq-\tr_{\Q}\left(\rho_{t}\int_{0}^{\infty}\left(x\1+\rho_{t+\varepsilon}\right)^{-1}\left(\rho_{t}-\rho_{t+\varepsilon}\right)\left(x\1+\rho_{t}\right)^{-1}\d x\right),
$$
so that we desire to verify the boundedness of the limit $\lim\limits_{\varepsilon\to0^{+}}\frac{|T_{\rho_{t},\epsilon}|}{\varepsilon}$. In fact, by using similar arguments that in \cite[Example 6.2.31]{BratteliRobinson} we can show that
$$
\left|T_{\rho_{t},\epsilon}\right|\leq \Vert\rho_{t+\varepsilon}^{-1}\Vert_{\Q}\Vert\rho_{t+\varepsilon}-\rho_{t}\Vert_{\Q},
$$
and since $\rho_{t+\varepsilon}-\rho_{t}=\left(\e^{\varepsilon\L}-1\right)\e^{t\L}\rho$, we obtain 
$$
\lim_{\varepsilon\to0^{+}}\frac{\left|T_{\rho_{t},\epsilon}\right|}{\varepsilon}\leq \Vert\rho_{t}^{-1}\Vert_{\Q}\Vert\L\rho_{t}\Vert_{\Q},
$$
which, by the hypothesis of the Lemma ensures that $\lim\limits_{\varepsilon\to0^{+}}\frac{|T_{\rho_{t},\epsilon}|}{\varepsilon}$ is bounded. Then, by the Lebesgue's dominated convergence theorem, $\frac{\d}{\d t}\mS(t)\doteq\lim\limits_{\varepsilon\to0^{+}}\frac{\mS(\rho_{t+\varepsilon})-\mS(\rho_{t})}{\varepsilon}$ can be written as  
\begin{eqnarray*}
\frac{\d}{\d t}\mS(t)&=&-\tr_{\Q}\left(\L\left(\rho_{t}\right)\ln(\rho_{t})-\rho_{t}\int_{0}^{\infty}\left(x\1+\rho_{t}\right)^{-1}\L(\rho_{t})\left(x\1+\rho_{t}\right)^{-1}\d x\right)\\ 
&=&-\tr_{\Q}\left(\L\left(\rho_{t}\right)\ln(\rho_{t})-\L(\rho_{t})\right)\\ 
&=&-\tr_{\Q}\left(\L\left(\rho_{t}\right)\ln(\rho_{t})\right),
\end{eqnarray*}
where we use the representation identity $\int\limits_{0}^{\infty}\left(x\1+\rho_{t}\right)^{-1}\rho_{t}\left(x\1+\rho_{t}\right)^{-1}\d x=\1$ and the Liouvillean given by \eqref{eq:ferm_liou}, proving \eqref{eq:entropy_liou}. Now, by taking into account expression \eqref{eq:HSevenQ} we can write
$$ 
\frac{\d}{\d t}\mS(\rho_{t})=-\tr_{\Q}\left(\rho_{t}\L\left(\ln\left(\rho_{t}\right)\right)\right).
$$
Finally, one use Lemma \ref{lemma:Bruijn} and the Liouvillean given by Expression \eqref{eq:ferm_liou} in order to get
\begin{align*}
\frac{\d}{\d t}\mS(\omega_{t})=\mJ(\omega_{t}).
\end{align*}
\end{proof}
We are ready to state a \emph{Stam inequality} for fermion systems: 
\begin{lemma}[Stam Inequality]\label{lemma:stam}
Let $\mathbf{A}$ and $\mathbf{B}$ be two interacting fermion systems of the same size, $|\mathbf{A}|=|\mathbf{B}|=\Nc$, and suppose that both are described by the Gaussian states $\omega_{\mathbf{A}}$ and $\omega_{\mathbf{B}}$, respectively. For $\cal\doteq[0,1]$ take $\lambda\in\cal$ and define $\lambda_{\mathbf{A}}\doteq\lambda,\lambda_{\mathbf{B}}\doteq1-\lambda$. For the real numbers $\alpha,\beta\in\R$ we have the \emph{Stam inequality} 
$$
\eta^{2}\mJ\left(\omega_{\mathbf{I}}\right)\leq\alpha^{2}\mJ\left(\omega_{\mathbf{A}}\right)+\beta^{2}\mJ\left(\omega_{\mathbf{B}}\right),
$$
where $\eta\doteq\sqrt{\lambda_{\mathbf{A}}}\alpha+\sqrt{\lambda_{\mathbf{B}}}\beta$ and $\omega_{\mathbf{I}}^{(\lambda)}$ is the output Gaussian state associated to the density matrix \eqref{eq:out_gaussian}. 
\end{lemma}
\begin{proof}
In order to prove the Lemma we need some pivotal results, which are similar to that given in \cite{konSmi14,palmargio14}. In fact, one can proof Expressions \cite[Eq. (46)--(47)]{konSmi14} via the matrix \eqref{eq:matrixM}, and observing that the covariance matrix of the composite system $\mathbf{I}\equiv\mathbf{A}\cup\mathbf{B}$ is \cite{gubner06}
$$
\gamma\doteq 
\begin{pmatrix}
\gamma_{\mathbf{A}} & \gamma_{\mathbf{A},\mathbf{B}}\\
\gamma_{\mathbf{B},\mathbf{A}} & \gamma_{\mathbf{B}} 
\end{pmatrix},\qquad\text{with}\qquad \gamma^{\mathrm{t}}\gamma\leq1_{2\Nc},\gamma_{\mathbf{B},\mathbf{A}}=-\gamma_{\mathbf{A},\mathbf{B}}^{\mathrm{t}},
$$
such that the covariance of the reduced density matrix given by \eqref{eq:out_gaussian}, i.e., $\rho_{\mathbf{I}}^{(\lambda)}\doteq\tr_{\mathbf{B}}\left(\mathbb{U}_{\lambda}^{*}\left(\rho_{\mathbf{A}}\otimes\rho_{\mathbf{B}}\right)\mathbb{U}_{\lambda}\right)$ is $\gamma_{\mathbf{I}}^{(\lambda)}=\lambda_{\textbf{A}}\gamma_{\mathbf{A}}+\lambda_{\textbf{B}}\gamma_{\mathbf{B}}-\sqrt{\lambda_{\textbf{A}}\lambda_{\textbf{B}}}(\gamma_{\mathbf{A},\mathbf{B}}+\gamma_{\mathbf{B},\mathbf{A}})$ whereas the displacement vector is (see Lemma \ref{lemma:displ_weyl})
$$
\mathrm{D}(\psi_{j})=\sqrt{\lambda_{\textbf{A}}}\B(\psi_{j})+\sqrt{\lambda_{\textbf{B}}}\mathrm{C}(\psi_{j}),\quad j\in J,
$$
c.f., \eqref{eq:dis_weyl_cliff}. Additionally, we desire to verify the compatibility of the Liouvillean \eqref{eq:ferm_liou} and the beam--splitter operator of Definition \ref{def:dis_weyl}, see also \eqref{eq:displa_clif}. Note that the Liouvillean \eqref{eq:ferm_liou} and the quantum channel expressed by the density matrix \eqref{eq:out_gaussian} are related by $\E_{\mathbb{U}_{\lambda}}\left(\e^{t_{\mathbf{A}}\L}\otimes\e^{t_{\mathbf{B}}\L}\right)(\rho_{\mathbf{A}}\otimes\rho_{\mathbf{B}})=\e^{t_{\mathbf{I}}\L}\E_{\mathbb{U}_{\lambda}}(\rho_{\mathbf{A}}\otimes\rho_{\mathbf{B}})$, such that the relation between covariance matrices is 
\begin{equation}
\lambda_{\textbf{A}}\tau_{t_{\textbf{A}}}\left(\gamma_{\textbf{A}}^{(\lambda)}\right)+\lambda_{\textbf{B}}\tau_{t_{\textbf{B}}}\left(\gamma_{\textbf{B}}^{(\lambda)}\right)=\tau_{t_{\textbf{I}}}\left(\gamma_{\textbf{I}}^{(\lambda)}\right)\qquad\text{with}\qquad\tau_{0}\left(\gamma_{\textbf{C}}^{(\lambda)}\right)=\gamma_{\textbf{C}}^{(\lambda)},\,\mathbf{C}\in\{\mathbf{A},\mathbf{B},\mathbf{I}\}.\label{eq:covABI}
\end{equation}
In above expressions, for $\mathbf{C}\in\{\mathbf{A},\mathbf{B},\mathbf{I}\}$, $\gamma_{\mathbf{C}},t_{\mathbf{C}}\in\R_{0}^{+}$ and $\tau_{t_{\textbf{C}}}\left(\gamma_{\textbf{C}}^{(\lambda)}\right)$ denote the covariance matrix, the time associated to the system $\mathbf{C}$ and the time \emph{automorphism} on the algebra of complex matrices $\mathrm{Mat}(N,\CP)$ of size $N\times N$, respectively. Note that $\tau_{0}\left(\gamma_{\textbf{C}}^{(\lambda)}\right)=\gamma_{\textbf{C}}^{(\lambda)}$. Expression \eqref{eq:covABI} physically means that the systems $\mathbf{A},\mathbf{B}$ and $\mathbf{I}$ evolve independently of each other. Thus we can take, for example, that the times $t_{\mathbf{C}}\in C(\R_{0}^{+};\R_{0}^{+})$ are continuous functions of a \emph{universal} time $t\in\R_{0}^{+}$, such that $t_{\mathbf{C}}(0)=0$. By \eqref{eq:covABI}, we observe that $\gamma_{\textbf{I}}^{(\lambda)}=\lambda_{\textbf{A}}\gamma_{\textbf{A}}+\lambda_{\textbf{B}}\gamma_{\textbf{B}}$, and denoting by $\tau_{t_{\mathbf{C}}(t)}(\gamma_{\mathbf{C}})\equiv\tau_{t}(\gamma_{\mathbf{C}})$, we find that the time evolution of $\gamma_{\textbf{C}}$ is
$$
\lambda_{\textbf{A}}\tau_{t}(\gamma_{\mathbf{A}})+\lambda_{\textbf{B}}\tau_{t}(\gamma_{\mathbf{B}})=\tau_{t}(\lambda_{\textbf{A}}\gamma_{\textbf{A}}+\lambda_{\textbf{B}}\gamma_{\textbf{B}}),
$$
which means that $\tau_{t}\colon\mathrm{Mat}(|\J|,\CP)\to\mathrm{Mat}(|\J|,\CP)$ is an \emph{affine transformation}, i.e., the \emph{time evolution} of $\gamma\in\mathrm{Mat}(|\J|,\CP)$ it must have the form $\tau_{t}(\gamma)=\gamma+tS$, where $\gamma,S\in\mathrm{Mat}(|\J|,\CP)$ are both skew--symmetric matrices, c.f. \cite[Lemmata III.1--III.2]{konSmi14}. Returning to \eqref{eq:covABI} we deduce that the compatibility of times is described by 
\begin{equation}
t_{\textbf{I}}(t)=\lambda_{\mathbf{A}}t_{\mathbf{A}}(t)+\lambda_{\mathbf{B}}t_{\mathbf{B}}(t)\quad\text{with}\quad t\in\R_{0}^{+}.\label{eq:timesIAB} 
\end{equation}
What is missing to verify is the compatibility of the displacement operator and the quantum channel: 
$$
\E_{\mathbb{U}_{\lambda}}\left(\left(\rho_{\mathbf{A}}\right)_{R_{\j}}^{(\alpha\theta)}\otimes\left(\rho_{\mathbf{B}}\right)_{R_{\j}}^{(\beta\theta)}\right)=\left(\rho_{\mathbf{I}}\right)_{R_{\j}}^{(\eta\theta)}=\left(\E_{\mathbb{U}_{\lambda}}(\rho_{\mathbf{A}}\otimes\rho_{\mathbf{B}})\right)_{R_{\j}}^{(\eta\theta)}
$$
for $\alpha,\beta,\theta\in\R$, some parameter $\eta$ depending of $\alpha$ and $\beta$, and for $\mathbf{C}\in\{\mathbf{A},\mathbf{B},\mathbf{I}\}$, $\left(\rho_{\mathbf{C}}\right)_{R_{\j}}^{(\theta)}$ is given by \eqref{eq:displa_clif2} for any self--adjoint element $R_{\j}$ of $\Q$, the Clifford algebra, and $\j\in\J$ defined by \eqref{eq:finite_setJ}. Similar to \cite{palmargio14} by \eqref{eq:displa_clif}, one obtain that $\eta=\sqrt{\lambda_{\mathbf{A}}}\alpha+\sqrt{\lambda_{\mathbf{B}}}\beta$. Then, straightforward calculations combinating Corollary \ref{corollary:Gaussian} and Lemmata \ref{lemma:qfi}, \ref{lemma:Bruijn} yield us to the inequality: $\eta^{2}\mJ\left(\omega_{\mathbf{I}}\right)\leq\alpha^{2}\mJ\left(\omega_{\mathbf{A}}\right)+\beta^{2}\mJ\left(\omega_{\mathbf{B}}\right)$. For complete details see \cite{palmargio14}.
\end{proof}
\appendix
\section{Open Systems, Quantum Channels and \eqt{$\C$}--algebras}\label{appen:qin}
If the physical system $\mathbf{S}$ interacts with its environment $\mathbf{E}$, \emph{dissipation processes} can occur (e.g. heat flow, diffusion processes, etc.).
To tackle these kind of problems one needs results and methods from \emph{quantum dynamical semigroups}. In this appendix we employ basic operator algebras in order to provide an introductory setting in the scope of open quantum systems and quantum information theory. In turn, we introduce relevant definitions and expressions used throughout the entire paper.
\subsection*{Completely positive maps}
Let $\H_{1}$ and $\H_{2}$ be two finite Hilbert spaces, with dimensions $\dim\H_{1}=m$ and $\dim\H_{2}=n$. The \emph{tensorial product} $\H_{1}\otimes\H_{2}$ is a Hilbert space such that for each $\varphi_{1}\in\H_{1}$ and each $\varphi_{2}\in\H_{2}$, $\varphi_{1}\otimes\varphi_{2}$ is a bilinear form. If $\{\psi_{1,i}\}_{i\in I}$ and $\{\psi_{2,i'}\}_{i'\in I'}$ are basis of $\H_{1}$ and $\H_{2}$ respectively, then $\{\psi_{1,i}\otimes\psi_{2,i'}\}_{(i,i')\in I\times I'}$ is a basis of $\H_{1}\otimes\H_{2}$. In particular, the Hilbert space $\H_{1}\otimes\H_{2}$ has dimension $|I|\times|I'|=mn$, its inner product is given by
$$
\inner{\zeta_{1}\otimes\zeta_{2},\varphi_{1}\otimes\varphi_{2}}_{\H_{1}\otimes\H_{2}}=\inner{\zeta_{1},\varphi_{1}}_{\H_{1}}\inner{\zeta_{2},\varphi_{2}}_{\H_{2}}, \qquad\zeta_{1},\varphi_{1}\in\H_{1},\,\zeta_{2},\varphi_{2}\in\H_{2},
$$
and for any $A\in\BL(\H_{1}), B\in\BL(\H_{2})$, $A\otimes B\in\BL(\H_{1}\otimes\H_{2})$, such that for any $\varphi_{1}\in\H_{1}$ and $\varphi_{2}\in\H_{2}$ is satisfied
$$
(A\otimes B)(\varphi_{1}\otimes\varphi_{2})=A\varphi_{1}\otimes B\varphi_{2}.
$$
Note that the Hilbert spaces $\H_{1}$ and $\H_{2}$ are isomorphic to the vectorial spaces $\CP^{m}$ and $\CP^{n}$ respectively, whereas $\H_{1}\otimes\H_{2}$ is isomorphic to $\CP^{mn}$. As is usual, for $N\in\N$, $\BL(\CP^{N})$ denotes the $^{*}$--algebra of all the continuous linear transformations of $\CP^{N}$ to $\CP^{N}$. Hence, if $N=mn$, $\BL(\CP^{m})$ and $\BL(\CP^{n})$ will be defined as sub $^{*}$--algebras given by $A_{1}\otimes\mathbf{1}_{\mathrm{Mat}(n,\CP)}$ and $\mathbf{1}_{\mathrm{Mat}(m,\CP)}\otimes A_{2}$. Here, for any $N\in\N$, $\mathbf{1}_{\mathrm{Mat}(N,\CP)}$ denotes the identity map on the set of complex matrices $\mathrm{Mat}(N,\CP)$ of size $N\times N$. With this notation we are able to define:
\begin{definition}[Completely positive maps]\label{def:comp_pos:map}
Let $m,n\in\N$ be two positive natural numbers. We say that the map $\Phi\colon\BL(\CP^{m})\to\BL(\CP^{n})$ is ``positive'' or ``positivity preserving operator'' if $A\geq0$\footnote{An operator $A\in\BL(\CP^{m})$ is positive if it is self--adjoint and $\spec(A)\geq0$.} implies that $\Phi(A)\geq0$. Additionally, if for any $n\in\N$
$$
\Phi\otimes\mathbf{1}_{\mathrm{Mat}(n,\CP)}\colon\BL(\CP^{m})\otimes\BL(\CP^{n})\to\BL(\CP^{m})\otimes\BL(\CP^{n}).
$$ 
is a positivity preserving operator, we say that $\Phi$ is \emph{completely positive}.  
\end{definition}
An important example of completely positive maps are partial traces, which are canonically introduced as follows:\\
The partial trace over $\BL(\CP^{n})$ is the unique linear transformation $\Tr_{\CP^{n}}\colon\BL(\CP^{mn})\to\BL(\CP^{m})$ satisfying
$$
\Tr_{\CP^{mn}}((A_{1}\otimes\mathbf{1}_{\mathrm{Mat}(n,\CP)})A_{2})=\Tr_{\CP^{m}}(A_{1}\Tr_{\CP^{n}}(A_{2})).
$$
Similarly, the partial trace over $\BL(\CP^{m})$ is the unique linear transformation $\Tr_{\CP^{m}}\colon\BL(\CP^{mn})\to\BL(\CP^{n})$ such that
$$
\Tr_{\CP^{mn}}(A_{1}(\mathbf{1}_{\mathrm{Mat}(m,\CP)}\otimes A_{2}))=\Tr_{\CP^{n}}(\Tr_{\CP^{m}}(A_{1})A_{2}).
$$
We say that an algebra $\W\equiv(\W,+,\cdot,^{*},\Vert\cdot\Vert_{\W})$ is a $\C$--algebra if it is equipped with an involution $^{*}$, it is complete, and it is endowed with a norm $\Vert\cdot\Vert_{\W}$ satisfying $\Vert A^{*}A\Vert_{\W}=\Vert A\Vert_{\W}^{2}$ for all $A\in\W$. $\W$ is said ``unital'' if it is embedded with a unit or identity operator, $\1$. If $\W_{1}$ and $\W_{2}$ are two finite $\C$--algebras isomorphic to $\BL(\CP^{m})$ and $\BL(\CP^{n})$ respectively, then the \emph{product $\C$--algebra} $\W\equiv\W_{1}\otimes\W_{2}$ is isomorphic to $\BL(\CP^{N})$, for $N=mn$. Here, an operator $A_{1}\in\W_{1}$ is view as an operator $A_{1}\otimes\mathbf{1}_{\mathrm{Mat}(n,\CP)}\in\W$, while, the operator $A_{2}\in\W_{2}$ is recognized as the operator $\mathbf{1}_{\mathrm{Mat}(m,\CP)}\otimes A_{2}\in\W$.
\subsection*{Quantum dynamical semigroups}
Consider a unital $\C$--algebra $\W$, and we denote its norm by $\Vert\cdot\Vert_{\W}$, we define:
\begin{definition}[Quantum dynamical semigroup]\label{def:semigroup}
Let $\W_{1}\subset\W$ be a unital subalgebra of $\W$. A semigroup on $\W_{1}$ is understood as a family $\P\doteq\{\P_{t}\}_{t\in\R_{0}^{+}}\in\BL(\W_{1})$ such that for any $s,t\in\R_{0}^{+}$ we have
\begin{equation}\label{eq:semigroup}
\P_{s}\P_{t}=\P_{s+t}, \qquad\text{and}\qquad\P_{0}=\1.
\end{equation}
$\P$ is a ``quantum dynamical semigroup (QDS)'' or a ``quantum Markov semigroup'' if 
$\P$ is a strongly continuous semigroup (or $C_{0}$--semigroup), i.e., if it is continuous in the strong operator topology, that is, for any $A\in\W_{1}$ one has $\lim\limits_{t\to0}\Vert\P_{t}A-A\Vert=0$. 
\end{definition}
One can check that for each $t\in\R_{0}^{+}$ exists constants $C\in\R$ and $D\geq1$ such that \cite{EngelNagel}
$$
\Vert\P_{t}\Vert_{\W_{1}}\leq D\e^{Ct}.
$$
Note that if $C=0$, then $\P$ is bounded. Additionally, if $C=0$ and $D=1$ then $\P$ is said ``contractive'' or is called a ``semigroup of contractions''. The semigroup is said to be an isometry if for any $A\in\W_{1}$ and $t\in\R_{0}^{+}$, we have $\Vert A\P_{t}\Vert_{\W_{1}}=\Vert A\Vert_{\W_{1}}$. Finally, if for any $N\in\N$ and $t\in\R_{0}^{+}$, $\P_{t}\otimes\mathbf{1}_{\mathrm{Mat}(N,\CP)}$ is a positivity preserving operator, we say that $\P$ is \emph{completely positive}. See Definition \ref{def:comp_pos:map}. In the latter case, for any fix $t$ in compact intervals, $\P_{t}$ is named a quantum dynamical map.\\
We naturally introduce the exponential function on $\W$ by
\begin{equation}
\e^{A}\doteq\1+\sum_{n=1}^{\infty}\frac{A^{n}}{n!},\qquad A\in\W,\label{eq:expCAR} 
\end{equation}
which is an absolutely convergent series on $\W$ if and only if there is $C_{A}\in\R_{0}^{+}$, such that for any $n\in\N$, 
$$
\Vert A\Vert_{\W}\leq C_{A}^{n}.
$$ 
Note that for any bounded operators $A,B\in\W$, we are able to apply the \emph{operator expansion theorem} 
\begin{equation}
\e^{B}A\e^{-B}=A+[B,A]+\frac{1}{2!}[B,[B,A]]+\cdots\equiv A+\sum_{n=1}^{\infty}\frac{\text{ad}_{B}^{n}(A)}{n!},\label{eq:oper_expCAR}
\end{equation}
where for $n\in\N$, $\text{ad}_{B}^{n}(A)\doteq[B,[B,[\ldots,A]]\ldots]]\in\W$ is the $n$--fold commutator of $A$ with $B$. Furthermore, for $A,B,C,D\in\W$ the identity
\begin{equation}
[AB,CD]=A\{B,C\}D-\{A,C\}BD+CA\{B,D\}-C\{A,D\}B,\label{eq:iden_fourtermsCAR}
\end{equation}
holds. Here, $[A,B]\doteq AB-BC\in\W$, $\{A,B\}\doteq AB+BA\in\W$ are the usual ``commutator'' and ``anticommutator'' operators on $\W$ of $A$ with $B$, respectively. Note that the series expressed in \eqref{eq:oper_expCAR} is absolutely convergent if there is $D_{A,B}\in\R_{0}^{+}$, such that for any $n\in\N$, 
$$
\Vert\text{ad}_{B}^{n}(A)\Vert_{\W}\leq D_{A,B}^{n}.
$$
Note that for $\C$--algebras, one can invoke the \emph{Campbell--Baker--Hausdorff Theorem}. In fact, it is a standard procedure shows that for $A,B\in\W$ such that $[A,[A,B]]=[B,[A,B]]=0$, we have
\begin{equation}\label{eq:CBHT}
\e^{A+B}=\e^{A}\e^{B}\e^{-\frac{1}{2}[A,B]}. 
\end{equation}
Concerning quantum dynamical maps, we know that for any $\P_{t}$ we are able to write it in terms of its infinitesimal generator $\L\in\W$ (a posibly unbounded operator on $\W$) or \emph{Liouvillean} as $\P_{t}=\e^{t\L}$ such that for any $A\in\W_{1}$ and $t\in\R_{0}^{+}$ we have
\begin{equation}
\frac{\d}{\d t}A_{t}=\L A_{t},\qquad A_{t}\doteq\e^{t\L}A.\label{eq:Liouvillean} 
\end{equation}
If $\W_{1}$ is a finite $\C$--algebra, the (bounded) Liouvillean is explicitly given by the \emph{Lindblad form}, that is, \cite{alickiLendi}
\begin{equation}
\L A=\ii[H,A]+\sum_{i\in I}V_{i}^{*}[A,V_{i}]+[V_{i}^{*},A]V_{i},\label{eq:Lindblad}
\end{equation}
where $I$ is an index set, $H\in\W_{1}$ is a self--adjoint operator known as the \emph{Hamiltonian} of the open quantum system described by $\W_{1}$, and $V_{i}\in\W$.
For unbounded $\L$, there is not such explicit form for this, see \cite{alickiLendi} for further details.
\subsection*{States}
Consider a unital separable $\C$--algebra $\W$. A linear functional $\omega\in\W^{*}$ is a ``state'' if it is positive and normalized, i.e., if for all $A\in\W,\omega(A^{*}A)\geq0$ and $\omega(\1)=1$. In the sequel, $\states_{\W}\subset\W^{*}$ will denote the set of all states on $\W$. Note that any $\omega\in\states_{\W}$ is \emph{Hermitian}, i.e., for all $A\in\W,\omega(A^{*})=\overline{\omega(A)}$. $\omega\in\states_{\W}$ is said to be ``faithful'' if $A=0$ whenever $A\geq 0$ and $\omega(A)=0$. The set of all elements $A\in\W$ such that $\spec(A)>0$ will we named the positive elements of $\W$ and it will denoted by $\W^{+}$. Note that by the \emph{Banach--Alaoglu Theorem}, $\states_{\W}$ is a compact set in the weak$^{*}$--topology $\sigma(\W^{*},\W)$. Moreover, since $\W$ is unital, under the topology $\sigma(\W^{*},\W)$, the set $\states_{\W}$ is a convex set, and its \emph{extremal} points coincide with the \emph{pure} states \cite[Theorem 2.3.15]{BratteliRobinsonI}. The latter, combining with the fact that $\W$ is separable allows to claim that the set of states $\states_{\W}$ is metrizable in $\sigma(\W^{*},\W)$ \cite[Theorem 3.16]{rudin}. Note that the existence of extremal points is a consequence of the \emph{Krein--Milman Theorem}. More specifically, if $\mathtt{E}(\states_{\W})$ denotes the set of extremal points of $\states_{\W}$,
$$
\states_{\W}=\mathrm{cch}\left(\mathtt{E}\left(\states_{\W}\right)\right),
$$
where, for $\X$ a Topological Vector Space and $A\subset\X$, $\mathrm{cch}(A)$ refers to the \emph{closed convex hull} of $A$. Such extremal points $\mathtt{E}(\states_{\W})$ or pure states can not be written as a linear combination of any states. As an application of the extremal states is that these are used to write any ``mixed state'' $\omega\in\states_{\W}$. By a mixed state $\omega\in\states_{\W}$, we mean that, there are  states $\{\omega_{j}\}_{j=1}^{m}\in\mathtt{E}(\states_{\W})$, $m\in\N$, and positive real numbers, $0\leq\lambda_{j}\leq1$ for $j\in\{1,\ldots,m\}$, with $\sum\limits_{j=1}^{m}\lambda_{j}=1$ satisfying
\begin{equation}\label{eq:convex_state}
\omega=\sum_{j=1}^{m}\lambda_{j}\omega_{j}.
\end{equation}
In particular, if the state $\omega\in\states_{\W}$ is pure, $\omega=\sum\limits_{j=1}^{m}\lambda_{j}\omega_{j}$ implies that $\omega=\omega_{1}=\cdots=\omega_{m}$, and $\lambda_{1}=\cdots=\lambda_{j}=\frac{1}{m}$. More generally, by the \emph{Choquet's Theorem}, the extremal points of $\states_{\W}$ form a Baire set $G_{\delta}$, and for any $\omega\in\states_{\W}$ there is a probability measure $\mu\in M_{+,1}(\states_{\W})$ supported by $\mathtt{E}(\states_{\W})$ with barycenter $\omega$ (i.e., for all $f\in\cA(\states_{\W})$, $\omega(f)=f(\omega)$)\footnote{If $C_{\R}(\states_{\W})$ is the set of all real continuous functions on $\states_{\W}$, then
$\cA(\states_{\W})\doteq\{f\in C_{\R}(\states_{\W});\, f(\lambda\omega_{1}+((1-\lambda)\omega_{2}))=\lambda f(\omega_{1})+(1-\lambda)f(\omega_{2}),\text{ with }\lambda\in[0,1]\}$ denotes of all the real continuous affine functions on $\states_{\W}$.} and satisfying $\mu(\mathtt{E}(\states_{\W}))=1$, \cite{phelpslectures,israel1979convexity}. Concretely, for all $f\in\cA(\states_{\W})$ and any $\omega\in\states_{\W}$, there is $\mu\in M_{+,1}(\states_{\W})$ such that
$$
f(\omega)=\int_{\mathtt{E}(\states_{\W})}f(\omega')\d\mu_{\omega}(\omega').
$$
In this paper we will deal with faithful--normal states. By normal we mean that these are determined by density matrices, i.e., operators $\rho\in\W^{+}\cap\W$ that satisfy $\tr_{\W}(\rho)=1$, where $\tr_{\W}\in\states_{\W}$ is the so--called ``tracial state'', which is the ``normalized trace'' on $\W$, i.e., $\tr_{\W}(A)\doteq\Tr_{\W}(A)/\tr_{\W}(\mathfrak{1})$, so that $\Tr_{\W}\in\W^{*}$ is the ``trace'' on $\W$, satisfying $\Tr_{\W}(AA^{*})=\Tr_{\W}(A^{*}A)$, for all $A\in\W$. More precisely, for any normal state $\omega\in\states_{\W}$, there exists a unique positive operator $\rho_{\omega}\in\W^{+}\cap\W$, with $\tr_{\W}(\rho_{\omega})=1$, such that the expectation value of $A\in\W$ w.r.t. $\rho_{\omega}$ is 
\begin{equation}\label{eq:tra_state}
\inner{A}_{\rho}\doteq\omega(A)=\tr_{\W}(\rho_{\omega}A). 
\end{equation}
As a trivial case, note that tracial state $\tr_{\W}\in\states_{\W}$ is faithful and normal. Note that for any density matrices $\rho,\rho_{1},\rho_{2}\in\W^{+}\cap\W$ we have the following two \emph{functional calculus} representation identities
\begin{align}
\ln\left(\rho_{2}\right)-\ln\left(\rho_{1}\right)&=\int_{0}^{\infty}\left(\left(x\1+\rho_{1}\right)^{-1}-\left(x\1+\rho_{2}\right)^{-1}\right)\d x\nonumber\\
&=
\int_{0}^{\infty}\left(x\1+\rho_{1}\right)^{-1}\left(\rho_{2}-\rho_{1}\right)\left(x\1+\rho_{2}\right)^{-1}\d x\label{eq:func_ident1}
\end{align}
and
\begin{equation}
\ln\left(\rho\right)=\int_{0}^{\infty}\left(\left(x\1+\1\right)^{-1}-\left(x\1+\rho\right)^{-1}\right)\d x.\label{eq:func_ident2}
\end{equation}
Let now $\{w_{i}\}_{i\in I}$ and $\1$ be the generators of the $\C$--algebra. That is, any element $A\in\W$ can be written by
\begin{equation}\label{eq:exp_genC}
A=\sum_{n\in\N}\sum_{j_{1},\ldots,j_{n}\in\{+,-\}}\sum_{i_{1},\ldots,i_{n}\in I}v_{A}(i_{1},\ldots,i_{n})w_{i_{1}}^{j_{1}}\cdots w_{i_{n}}^{j_{n}}, 
\end{equation}
where $w_{i}^{-}\doteq w_{i}$ and $w_{i}^{+}\doteq w_{i}^{*}$, for $i\in I$, and $v_{A}(i_{1},\ldots,i_{n})\colon I^{n}\to\CP$, is an $I^{n}$ bounded complex function depending on $A$. Note that $A$ has not necessarily a unique form to be written. 
\begin{definition}[Gaussian states on $\C$--algebras]\label{def:gauss_states_Calg}
In \eqref{eq:exp_genC} take $n=2$ and $A\in\W$ invertible, i.e., there is a unique element $A^{-1}\in\W$ such that $AA^{-1}=A^{-1}A=\1$. We say that the state $\omega_{A}\in\states_{\W}$ is a Gaussian state associated to $A$ if and only if its density matrix $\rho_{A}\in\W^{+}\cap\W$ can be uniquely written as $\rho_{A}\doteq\frac{\e^{\alpha A}}{\tr_{\W}\left(\e^{\alpha A}\right)},\alpha\in\CP. $ For $M\in\R^{+}$, the operator $g_{A}\doteq M\e^{\alpha A}\in\W^{+}\cap\W$ is called a ``Gaussian operator'' associated to $A$. The set of all Gaussian states associated to $A$ will be denoted by $\states_{\W,A}$, whereas $\W_{A}$ will denote the set of all Gaussians operators.
\end{definition}
\subsection*{Channels: Completely positive and trace preserving maps}
If a physical system $\mathbf{A}$ is interacting (or coupled) with another one $\mathbf{B}$, we usually assume that these are described by $\C$--algebras, namely, $\W_{\mathbf{A}}$ and $\W_{\mathbf{B}}$, respectively. Therefore, the interacting system $\mathbf{I}\equiv\mathbf{A}\cup\mathbf{B}$ is described by the product $\C$--algebra $\W_{\mathbf{I}}\equiv\W_{\mathbf{A}}\otimes\W_{\mathbf{A}}$. Then, for any states $\omega_{\textbf{A}}\in\states_{\W_{\mathbf{A}}}$ and $\omega_{\textbf{B}}\in\states_{\W_{\mathbf{B}}}$ the interacting state $\omega_{\textbf{I}}\in\states_{\W_{\mathbf{I}}}$ is explicitly given by $\omega_{\textbf{I}}=\omega_{\textbf{A}}\otimes\omega_{\textbf{B}}$. If the states are normal states, it follows that their associated density matrices are given by $\rho_{\textbf{I}}=\rho_{\textbf{A}}\otimes\rho_{\textbf{B}}$, so that 
$$
\omega_{\textbf{I}}(A)=\tr_{\W_{\textbf{I}}}\left(\left(\rho_{\textbf{A}}\otimes\rho_{\textbf{B}}\right)A\right),\qquad A\in\W_{\textbf{I}}.
$$
Let $\mathcal{F}_{\mathbb{U}}\colon\W\to\W$ be the $^{*}$--automorphism on $\W$ defined uniquely by
$$
\mathcal{F}_{\mathbb{U}}(A)\doteq\mathbb{U}^{*}A\mathbb{U},\quad A\in\W.
$$
Here, $\mathbb{U}\in\BL(\W)$ is a unitary bounded operator on $\W$, which implemented the $^{*}$--automorphism $\mathcal{F}_{\mathbf{U}}$. Thus if $\mathbf{I}\equiv\mathbf{A}\cup\mathbf{B}$ is the interacting system mentioned above, then the operation $\E_{\mathbb{U}}\colon\W_{\mathbf{A}}\otimes\W_{\mathbf{B}}\to\W_{\mathbf{A}}$ defined by 
\begin{equation}
\E_{\mathbb{U}}(\omega_{\mathbf{I}})\doteq\tr_{\mathbf{B}}(\mathcal{F}_{\mathbb{U}}(\rho_{\mathbf{I}}))\equiv\tr_{\mathbf{B}}\left(\mathbb{U}^{*}\left(\rho_{\mathbf{A}}\otimes\rho_{\mathbf{B}}\right)\mathbb{U}\right)\label{eq:map_red_sys}
\end{equation}
describes the \emph{reduced system} $\mathbf{A}$ \cite{niel10}, where $\tr_{\mathbf{B}}$ is the normalized partial trace over $\W_{\mathbf{B}}$. In this paper we will say that the $^{*}$--automorphism $\mathcal{F}_{\mathbb{U}}$ is ``trace preserving'' if it is valid the condition
\begin{equation}
\mathbb{U}\mathbb{U}^{*}=\mathfrak{1},\label{eq:trace_preser}
\end{equation}
In particular, we say that the map $\E_{\mathbb{U}}$ is a ``quantum channel'' with inputs $\mathbf{A},\mathbf{B}$ and output $\mathbf{C}$ if it is completely positive and is trace preserving (CPTP). Recall that the partial trace is a completely positive map, and hence $ \E_{\mathbb{U}}$ is a well--defined quantum channel.\\

\noindent \textit{Acknowledgments:} This work is supported by \emph{Departamento de Física} of \emph{Universidad de Los Andes}. We are very grateful to A. Vershynina and N. Datta for hints and discussions. 

\bibliography{books}
\bibliographystyle{amsalpha}
\end{document}